\documentclass[11pt,leqno]{article}
\usepackage{amssymb,amsfonts}
\usepackage{amsmath,amsthm,amsxtra}

\usepackage[all]{xy}
\usepackage{color}
\usepackage{verbatim}

\usepackage{enumerate}

\newcommand\bigcheck[1]{#1 \raise1ex\hbox{$\hspace{-1ex}{}^\vee$}}
\newcommand\sucheck[1]{#1 \raise0.5ex\hbox{$\hspace{-1ex}{}^\vee$}}

\makeatletter
\@addtoreset{equation}{section}
\makeatother

\newtheorem{theorem}{Theorem}[section]
\newtheorem{lemma}[theorem]{Lemma}
\newtheorem{corollary}[theorem]{Corollary}
\newtheorem{proposition}[theorem]{Proposition}
\newtheorem*{lemma*}{Lemma}

\theoremstyle{definition}

\theoremstyle{remark}
\newtheorem{remark}[theorem]{Remark}

\newcommand{\mc}[1]{{\mathcal #1}}
\newcommand{\mf}[1]{{\mathfrak #1}}
\newcommand{\mb}[1]{{\mathbb #1}}

\newcommand\tint{{\textstyle\int}}

\renewcommand{\tilde}{\widetilde}

\newcommand{\Mat}{\mathop{\rm Mat }}
\newcommand{\Der}{\mathop{\rm Der }}

\renewcommand{\ker}{\mathop{\rm Ker }}
\newcommand{\im}{\mathop{\rm Im }}

\newcommand{\Span}{\mathop{\rm Span }}

\newcommand{\codim}{\mathop{\rm codim }}

\definecolor{light}{gray}{.9}

\begin{document}


\title{On integrability of some bi-Hamiltonian two field systems of PDE}

\author{
Alberto De Sole
\thanks{Dipartimento di Matematica, 
Universit\`a di Roma ``La Sapienza'',
00185 Roma, Italy ~~
and 
Department of Mathematics, M.I.T.,
Cambridge, MA 02139, USA.~~
desole@mat.uniroma1.it ~~~~
Supported in part by PRIN and FIRB grants.
},~~
Victor G. Kac
\thanks{Department of Mathematics, M.I.T.,
Cambridge, MA 02139, USA.~~
kac@math.mit.edu~~~~
Supported in part by NSA grant.
}~~
and Refik Turhan 
\thanks{
Department of Engineering Physics,
Ankara University,
06100 Tandogan, Ankara
Turkey.~~
turhan@eng.ankara.edu.tr
}
}

\maketitle

\begin{abstract}
We continue the study of integrability of bi-Hamiltonian systems
with a compatible pair of local Poisson structures $(H_0,H_1)$,
where $H_0$ is a strongly skew-adjoint operator.
This is applied to the construction of some new two field integrable systems
of PDE by taking the pair $(H_0,H_1)$
in the family of compatible Poisson structures that arose in the study
of cohomology of moduli spaces of curves.
\end{abstract}

\section{Introduction}\label{sec:1}

The present paper is a continuation of our paper \cite{DSKT13},
where we began the study of integrability of bi-Hamiltonian PDE,
with a compatible pair of local Poisson structures $(H_0,H_1)$,
such that $H_0$ is strongly skew-adjoint.
Recall \cite{DSKT13} that a skew-adjoint $\ell\times\ell$ matrix differential operator $H$
is called \emph{strongly skew-adjoint} over an algebra of differential functions $\mc V$
if 
\begin{enumerate}[(i)]
\item
$\ker H\subset\delta\big(\mc V/\partial\mc V\big)$;
\item
$(\ker H)^\perp=\im H$,
\end{enumerate}
where $\delta$ is the variational derivative,
and $\perp$ denotes orthogonal complement with respect to the
bilinear form $\mc V^\ell\times\mc V^\ell\to\mc V/\partial\mc V$
defined by $(F,G)\mapsto\tint F\cdot G$.

First, we develop a theory of strongly skew-adjoint matrix differential operators
over a field of differential functions $\mc K$.
In particular we show that, in this case, we can replace condition (ii) by
\begin{enumerate}[(i')]
\setcounter{enumi}{1}
\item
$\dim_{\mc C}(\ker H)=\deg(H)$,
\end{enumerate}
where $\deg(H)$ is the degree of the Dieudonn\'e determinant of $H$,
and $\mc C\subset\mc K$ is the subfield of constants
(see Corollary \ref{20140407:prop1}).
This equivalent definition is more convenient if one is willing to work over a field of differential functions.

Second, we use the approach to the Lenard-Magri scheme of integrability
developed in \cite[Prop.2.9]{BDSK09}
to show that if $H_0$ is strongly skew-adjoint and
\begin{equation}\label{20140502:eq1}
C(H_0)\cap C(H_1)
\,\,\text{ has codimension } 1 \text{ in }\,\, C(H_0)
\,,
\end{equation}
where $C(H)$ denotes the space of Casimirs of $H$,
then the Lenard-Magri scheme always works beginning 
with $\tint h_0\in C(H_0)\backslash C(H_1)$
(see Corollary \ref{20140409:cor} and Remark \ref{20140409:rem}).
(Note that condition \eqref{20140502:eq1}
holds for all $(H_0,H_1)$ considered in the present paper,
but fails for the pair $(H_0,H_1)$ studied in \cite{DSKT13},
which makes the proof of integrability for the latter pair more difficult.)

We apply these results to arbitrary pairs $(H_0,H_1)$
from a $6$-parameter family of pairwise compatible
Poisson structures that arose in the study of the second cohomology
of the moduli spaces of curves.
This $6$-parameter family naturally corresponds to the Lie conformal algebra $\hat{L}$,
which is the central extension by the space of all $2$-cocycles,
with values in the base field $\mb F$,
of the Lie conformal algebra 
$$
L=\mb F[\partial]u\oplus\mb F[\partial] v\,,
$$
with $\lambda$-brackets
$$
[u_\lambda u]=(\partial+2\lambda) u
\,\,,\,\,\,\,
[u_\lambda v]=(\partial+\lambda) v
\,\,,\,\,\,\,
[v_\lambda v]=0
\,.
$$
The Lie conformal algebra $L$ in turn corresponds to the 
Lie algebra $\mf g$ of differential operators of order at most $1$ on the circle.
It was a key observation of \cite{ADCKP88}
that $H^2(\mf g,\mb F)$ is canonically isomorphic to the second
cohomology of the moduli space of curves of given genus $g\geq5$.

A systematic analysis of all pairs $(H_0,H_1)$
coming from the $6$-parameter family produces
several large families of integrable bi-Hamiltonian PDE
(see Sections \ref{sec:4} and \ref{sec:5}).
Some special cases of these PDE are well-known integrable systems
\cite{Ito82,Kup85a,Kup85b,AF88,FL96,GN90}.
Other special cases seem to be new,
like 
equation \eqref{6.1}, which is an extension of the Antonowicz-Fordy equation \cite{AF88};
equation \eqref{6.2}, which is an extension of an equation which appeared in the list of \cite{MSY87};
equation \eqref{turhan:eq}, which is an extension of 
the Fokas-Liu equation \cite{FL96}, the Kupershmidt equation \cite{Kup85b},
the Ito equation \cite{Ito82}, and the Kaup-Broer equation \cite{Kup85a};
and equations \eqref{turhan:eqb}, \eqref{6.5} and \eqref{6.6}.
Upon potentiation, some special cases of these equations
also appeared in the list of \cite{MSY87}.
All of this is discussed in Section \ref{sec:6}.

Throughout the paper all vector spaces are considered over a field $\mb F$ of characteristic zero.

We wish to thank Vladimir Sokolov for enlightening correspondence
and Pavel Etingof for bright observations.

\section{The Lenard-Magri scheme of integrality}\label{sec:2}

\subsection{Algebras of differential functions}\label{sec:2.1}

Consider
the algebra of differential polynomials
$R_\ell=\mb F[u_i^{(n)}\,|\,i\in I,n\in\mb Z_+]$,
with the derivation $\partial$ defined on generators by $\partial u_i^{(n)}=u_i^{(n+1)}$.
One has:
\begin{equation}\label{eq:comm-rel}
[\frac{\partial}{\partial u_i^{(n)}},\partial]=\frac{\partial}{\partial u_i^{(n-1)}}
\,\,\text{ for every }\,\,
i\in I,\,n\in\mb Z_+
\,.
\end{equation}
(the RHS is considered to be 0 for $n=0$).
Recall 
that an \emph{algebra of differential functions} $\mc V$ in the variables $u_i,\,i\in I=\{1,\dots,\ell\}$,
is a differential algebra extension of $R_\ell$,
endowed with commuting derivations $\frac{\partial}{\partial u_i^{(n)}}:\,\mc V\to\mc V$
extending the usual partial derivatives on $R_\ell$,
such that \eqref{eq:comm-rel} holds
and, for every $f\in\,\mc V$, we have $\frac{\partial f}{\partial u_i^{(n)}}=0$
for all but finitely many values of $i$ and $n$.
We will assume throughout the paper that $\mc V$ is a domain.

Consider the following filtration of the algebra of differential functions $\mc V$:
$$
\mc V_{m,i}\,=\,
\Big\{ f\in\mc V\,\Big|\,\frac{\partial f}{\partial u_j^{(n)}}=0
\text{ for all } (n,j)>(m,i)
\Big\}\,,
$$
where $>$ denotes lexicographic order.
%
%
The algebra of differential functions $\mc V$ is called \emph{normal} \cite{BDSK09}
if $\frac{\partial}{\partial u_i^{(m)}}(\mc V_{m,i})=\mc V_{m,i}$
for all $i\in I,m\in\mb Z_+$.
Note that any algebra of differential function can be extended to a normal one
(see \cite{DSK13}).
Examples of normal algebras of differential functions are
$R_\ell$, and $R_\ell[u_1^{-1},\log(u_1)]$.

We denote by $\mc F\subset\mc V$ the subagebra of quasiconstants:
$$
\mc F=\big\{\alpha\in\mc V\,\big|\,\frac{\partial \alpha}{\partial u_i^{(n)}}=0\,\,
\text{ for all }\,\, i\in I,n\in\mb Z_+\big\}
\,.
$$
Obviously $\mc F\subset\mc V$ is a differential subalgebra,
and we will assume, without loss of generality, that it is a differential field.
We also denote by $\mc C\subset\mc V$ the subagebra of constants:
$$
\mc C=\big\{a\in\mc V\,\big|\,\partial a=0\big\}
\,.
$$
It follows from \eqref{eq:comm-rel} that $\mc C$ is a subfield of $\mc F$.

\subsection{Local functionals, evolutionary vector fields, and variational derivative
}\label{sec:2.1b}

We call $\mc V/\partial\mc V$ the space of \emph{local functionals},
and we denote by $\tint f\in\mc V/\mc V$ the coset of $f\in\mc V$
in the quotient space.

An \emph{evolutionary vector field}
is a derivation of $\mc V$ of the form 
$$
X_P=\sum_{i\in I,n\in\mb Z_+}(\partial^n P_i)\frac{\partial}{\partial u_i^{(n)}}
\,\,,\,\,\,\,
P\in\mc V^\ell
\,.
$$ 
Evolutionary vector fields commute with $\partial$, they form a Lie subalgebra of $\Der(\mc V)$,
and the Lie bracket of two evolutionary vector fields is given by the formula 
$$
[X_P,X_Q]=X_{[P,Q]}
\,\,\text{ where }\,\,
[P,Q]=X_P(Q)-X_Q(P)
\,.
$$

The \emph{variational derivative} of $\tint f\in\mc V/\partial\mc V$ is
$\delta f=\big(\frac{\delta f}{\delta u_i}\big)_{i\in I}\in\mc V^\ell$, where
\begin{equation}\label{20130222:eq6}
\frac{\delta f}{\delta u_i}=
\sum_{n\in\mb Z_+}(-\partial)^n\frac{\partial f}{\partial u_i^{(n)}}\,.
\end{equation}

\subsection{Poisson structures}\label{sec:2.1c}

Given an $\ell\times\ell$-matrix differential operator
$H=\big(H_{ij}(\partial)\big)_{i,j\in I}\in\Mat_{\ell\times\ell}\mc V[\partial]$,
we associate the bracket
$\{\cdot\,,\,\cdot\}_H$ on $\mc V/\partial\mc V$,
given by
\begin{equation}\label{20130222:eq7}
\{\tint f,\tint g\}_H
=\tint \delta g\cdot H(\partial)\delta f\,.
\end{equation}
Note that the bracket \eqref{20130222:eq7} is skewsymmetric if and only if $H$ is skew-adjoint.
If \eqref{20130222:eq7} defines a Lie algebra bracket on $\mc V/\partial\mc V$,
then $H$ is called a \emph{Poisson structure} on $\mc V$.
Two Poisson structures $H_0$ and $H_1$ on $\mc V$
are called \emph{compatible} if their sum $H_0+H_1$
is again a Poisson structure.

Recall that 
a matrix differential operator $H\in\Mat_{\ell\times\ell}\mc V[\partial]$
is called \emph{non-degenerate} if 
it is not a left (or right) zero divisor in $\Mat_{\ell\times\ell}\mc V[\partial]$
(equivalently, if its Dieudonn\'e determinant in non-zero).

A \emph{Casimir element} for the Poisson structure $H$ is a local functional $\tint f\in\mc V/\partial\mc V$
such that $\delta f\in\ker(H)$,
i.e. the space of Casimir elements for $H$ is
\begin{equation}\label{20130222:eq9}
C(H)
=\Big\{\tint f\in\mc V/\partial\mc V\,\Big|\,H(\partial)\delta f=0\Big\}\,.
\end{equation}
\begin{remark}\label{20140408:rem}
Recall that if $H_0$ and $H_1$ are compatible Poisson structures,
then $C(H_0)$ is a Lie algebra with respect 
to the Lie bracket $\{\cdot\,,\,\cdot\}_1=\{\cdot\,,\,\cdot\}_{H_1}$ \cite{DSKT13}.
\end{remark}

\subsection{Strongly non-degenerate and strongly skew-adjoint matrix differential operators
over a differential field}\label{sec:2.1d}

For a matrix differential operator $H\in\Mat_{\ell\times\ell}\mc V[\partial]$,
we have $\im H^*\subset(\ker H)^\perp$,
where the orthogonal complement is with respect to the pairing 
$\mc V^\ell\times\mc V^\ell\to\mc V/\partial\mc V$ given by $(F,P)\mapsto\tint F\cdot P$.
We want to describe the situations when equality holds.

Throughout this section, unless otherwise specified, 
we consider instead of $\mc V$ its field of fractions $\mc K= Frac(\mc V)$.
Recall that any $\ell\times\ell$ matrix differential operator $H\in\Mat_{\ell\times\ell}\mc K[\partial]$
can be brought, by elementary row transformations,
to an upper triangular matrix $H_1\in\Mat_{\ell\times\ell}\mc K[\partial]$ 
(see e.g. \cite[Lem.3.1]{CDSK14}).
The matrix $H$ is non-degenerate if and only if $H_1$ has non-zero
diagonal entries.
In this case, the \emph{Dieudonn\`e determinant} of $H$ is
\begin{equation}\label{20140408:eq2}
\det(H)=\det{}_1(H)\xi^{\deg(H)}
\,,
\end{equation}
where $\det_1(H)\in\mc K$ is the product of the leading coefficients
of the diagonal entries of $H_1$,
and the \emph{degree} $\deg(H)$ of $H$ is the sum of the differential orders of the
diagonal entries of $H_1$.
If $H=AB$ is non-degenerate, then
\begin{equation}\label{20140408:eq2b}
\det(H)=\det(A)\det(B)
\,\,\text{ and }\,\,
\deg(H)=\deg(A)+\deg(B)
\,.
\end{equation}
Recall also that, if $H$ is non-degenerate, then 
$$
\dim_{\mc C}(\ker H)\leq\deg(H)
\,.
$$
(Note that we can always construct a differential field extension $\widetilde{\mc K}$ of $\mc K$
with the same field of constants
such that, in this extension, $\dim_{\mc C}(\ker H)=\deg(H)$,
see \cite[Lem.4.3]{CDSK13},
but, in general, this extension $\widetilde{\mc K}$ will not be an algebra of differential functions.)
For a more detailed exposition on the Dieudonn\`e determinant, see \cite{CDSK14}.
\begin{proposition}\label{20140407:prop2}
Let $H\in\Mat_{\ell\times\ell}\mc K[\partial]$ 
be a non-degenerate $\ell\times\ell$ matrix differential operator.
Then there exist 
non-degenerate matrices
$A,B\in\Mat_{\ell\times\ell}\mc K[\partial]$ 
such that 
\begin{equation}\label{20140407:eq1}
H=AB
\,\,,\,\,\,\,
\ker H=\ker B
\,\,\text{ and }\,\,
\deg(B)=\dim_{\mc C}(\ker B)
\,.
\end{equation}
\end{proposition}
\begin{proof}
Let $\{F_1,\dots,F_s\}$ be an ordered basis of $\ker H$, 
considered as a (finite dimensional) vector space over $\mc C$.
Let $F_1=\left(\begin{array}{l} f_1 \\ \vdots \\ f_\ell \end{array}\right)\in\ker H\backslash\{0\}$,
and assume, without loss of generality, that $f_\ell\neq0$.
We have the decomposition $H=H_1 K_1$
in $\Mat_{\ell\times\ell}\mc K[\partial]$, where
\begin{equation}\label{20140407:eq2}
K_1
=
\left(\begin{array}{llllc} 
1 & 0 & \dots & 0 & -\frac{f_1}{f_\ell} \\ 
0 & 1 & \dots & 0 & -\frac{f_2}{f_\ell} \\
\vdots&&\dots&&\vdots \\
0 & 0 & \dots & 1 & -\frac{f_{\ell-1}}{f_\ell} \\
0 & 0 & \dots & 0 & -\partial+\frac{f_\ell^\prime}{f_\ell}
\end{array}\right)
\,.
\end{equation}
(For scalar matrices, it is clear by the division algorithm.
For arbitrary matrices, the identity $H(\partial)F_1=0$
gives conditions on the entries of $H$
which are equivalent to saying that $H=H_1K_1$
for some $H_1\in\Mat_{\ell\times\ell}\mc K[\partial]$.)
Note that $\ker K_1=\mc CF_1\subset\ker H$, and $K_1(\ker H)\subset\ker H_1$.
Hence,
$K_1(\partial)F_2=\left(\begin{array}{l} g_1 \\ \vdots \\ g_\ell \end{array}\right)\in\ker H_1\backslash\{0\}$.
We can decompose, as before,
$H_1=H_2 K_2$
in $\Mat_{\ell\times\ell}\mc K[\partial]$, where
$K_2$ is as in \eqref{20140407:eq2} with $f_i$'s replaced by $g_i$'s.
Hence, $H=H_2K_2K_1$.
Repeating the same argument $s$ times,
we get
$$
H=H_sK_s\dots K_1
\,,
$$
where $K_1,\,\dots,K_s$ have all degree $1$,
and their kernels are
$$
\ker K_1=\mc C F_1
\,,\,\,
\ker K_2=\mc C K_1(\partial)F_2
,\,\dots\,,
\ker K_s=\mc C K_{s-1}(\partial)\dots K_1(\partial)F_s
\,.
$$
Hence, letting $B=K_s\dots K_1$,
we have that 
\begin{equation}\label{20140407:eq3}
\Span{}_{\mc C}\{F_1,\dots,F_s\}\subset\ker B
\,.
\end{equation}
On the other hand, we also have $\deg(B)=s$,
which implies that equality holds in \eqref{20140407:eq3}
\end{proof}
\begin{corollary}\label{20140407:prop1}
Let $H\in\Mat_{\ell\times\ell}\mc K[\partial]$ 
be a non-degenerate $\ell\times\ell$ matrix differential operator.
The following conditions are equivalent:
\begin{enumerate}[(a)]
\item
if $H$ is as in \eqref{20140407:eq1}, then $\deg(A)=0$
(i.e. $A$ is invertible in $\Mat_{\ell\times\ell}\mc K[\partial]$);
\item
$(\ker H)^\perp= \im H^*\subset\mc K^\ell$;
\item
$\dim_{\mc C}(\ker H)=\deg(H)$.
\end{enumerate}
\end{corollary}
\begin{proof}
By \eqref{20140407:eq1}, we have
$$
\deg(H)=\deg(A)+\deg(B)
\,\,\text{ and }\,\,
\dim_{\mc C}(\ker H)=\deg(B)
\,.
$$
It follows that $\deg(H)=\dim_{\mc C}(\ker H)$ if and only if $\deg(A)=0$,
i.e. conditions (a) and (c) are equivalent.

Next, let $\{F_1,\dots,F_s\}$ be a basis of $\ker H$,
and let $K_1(\partial),\dots,K_s(\partial)\in\Mat_{\ell\times\ell}\mc K[\partial]$
be as in the proof of Proposition \ref{20140407:prop2}.
We claim that
\begin{equation}\label{20140407:eq4}
\{F_1,\dots,F_s\}^\perp\subset\im K_1^*\dots K_s^*
\,.
\end{equation}
To prove this, let 
$P=\left(\begin{array}{l} p_1 \\ \vdots \\ p_\ell \end{array}\right)\in(\ker H)^\perp$.
We have
$\tint P\cdot F_1=0$, which is equivalent to
\begin{equation}\label{20140407:eq5}
p_\ell=-\frac{f_1}{f_\ell}p_1-\dots-\frac{f_{\ell-1}}{f_\ell}p_{\ell-1}
+\Big(\partial+\frac{f_\ell^\prime}{f_\ell}\Big)\widetilde{p_{\ell}}
\,,
\end{equation}
for some $\widetilde{p_{\ell}}\in\mc K$.
But equation \eqref{20140407:eq5}
is the same as saying that $P=K_1^*(\partial)P_1$,
where $P_1=\left(\begin{array}{l} p_1 \\ \vdots \\ p_{\ell-1} \\ \widetilde{p_\ell} \end{array}\right)$.
Next, we have:
$$
\tint P\cdot F_2=0=\tint (K_1^*(\partial)P_1)\cdot F_2=\tint P_1\cdot K_1(\partial)F_2
\,.
$$
Recalling that $\ker K_2=\mc CK_1(\partial)F_2$,
and repeating the same argument as above, we get $P_1=K_2^*(\partial)P_2$,
for some $P_2\in\mc K^\ell$.
Repeating the same argument $s$ times, we conclude that
$P=K_1^*(\partial)\dots K_s^*(\partial)P_s$
for some $P_s\in\mc K^\ell$.
This proves \eqref{20140407:eq4}.

The inclusion $(\ker H)^\perp\supset\im H^*$ in (b) always holds.
Assuming that condition (a) holds, 
we have, by \eqref{20140407:eq4}, 
$$
(\ker H)^\perp=(\ker B)^\perp\subset\im B^*=\im H^*
\,,
$$
since, by (a), $A$ is invertible in $\Mat_{\ell\times\ell}\mc K[\partial]$.
Hence, (a) implies (b).

Finally, we prove that (b) implies (a).
By (b), we have $(\ker H)^\perp=\im H^*$.
We let, as in \eqref{20140407:eq1}, $H=AB$,
where $A$ and $B$ are such that 
$\ker H=\ker B$ and $\dim_{\mc C}(\ker B)=\deg(B)$.
Hence, condition (c) holds for $B$,
and since (c) implies (b), we know that 
$(\ker B)^\perp=\im B^*$.
It follows by the above observations that
$$
\im H^*=B^*(\im A^*)=\im B^*
\,.
$$
In other words, by elementary linear algebra, we have
\begin{equation}\label{20140408:eq1}
\im A^*+\ker B^*=\mc K^\ell
\,.
\end{equation}
To conclude, we observe that, since $B$ is non-degenerate,
then $\ker B^*$ is finite dimensional over $\mc C$,
and therefore condition \eqref{20140408:eq1}
implies that $\deg(A)=0$, due to the following lemma.
\begin{lemma}\label{20140408:lem}
Any matrix differential operator $L\in\Mat_{\ell\times\ell}\mc K[\partial]$
which is non-degenerate and of degree $\deg(L)\geq1$ 
has image $\im L$ of infinite codimension over $\mc C$.
\end{lemma}
\begin{proof}
First, we can reduce to the scalar case. Indeed, there exist invertible matrices
$U,V\in\Mat_{\ell\times\ell}\mc K[\partial]$ 
and a diagonal matrix $D=$diag$(d_1,\dots,d_\ell)\in\Mat_{\ell\times\ell}\mc K[\partial]$ 
such that $L=UDV$,
see e.g. \cite[Lem.3.1]{CDSK14}.
But then 
$$
\codim{}_{\mc C}(\im L)=\codim{}_{\mc C}(\im D)
=\codim{}_{\mc C}(\im d_1)+\dots+\codim{}_{\mc C}(\im d_\ell)
\,.
$$
For a scalar differential operator $L(\partial)\in\mc K[\partial]$
of order greater than or equal to $1$,
the claim follows by the simple observation that
elements $(u^{(n)})^2$, for $n>>0$,
are linearly independent in $\mc K/(\im L)$
(this follows by \cite[Eq.(1.8)]{BDSK09}).
\end{proof}
\end{proof}
\begin{remark}\label{20140502:rem}
Proposition \ref{20140407:prop2} holds over any differential field $\mc K$,
while Corollary \ref{20140407:prop1} holds over any differential field $\mc K$
satisfying Lemma \ref{20140408:lem}.
\end{remark}
%
%
%
\begin{remark}\label{20140407:prop3}
It is not hard to show that,
if $H\in\Mat_{\ell\times\ell}\mc K\partial]$ is such that $\delta C(H)=\ker H$,
then $H$ is non-degenerate.
\end{remark}
\begin{proposition}\label{20140407:prop}
Any non-degenerate skew-adjoint matrix differential operator $H$
with quasiconstant coefficients is strongly skew-adjoint
over an algebra of differential functions of the form
$\widetilde{\mc V}=\widetilde{\mc F}\otimes_{\mc F}\mc V$,
where $\widetilde{\mc F}$ is some differential field extension over $\mc F$
with the same subfield of constants,
and partial derivatives $\frac{\partial}{\partial u_i^{(n)}}$ extended from $\mc V$
to $\widetilde{\mc V}$ acting trivially on $\widetilde{\mc F}$.
\end{proposition}
\begin{proof}
By assumption $H\in\Mat_{\ell\times\ell}\mc F[\partial]$.
Hence, $\ker H\subset\mc F^\ell$,
and there exists 
a differential field extension $\widetilde{\mc F}$ of $\mc F$
with the same subfield of constants $\mc C$
such that $\ker H$ in $\widetilde{\mc F}^\ell$ has dimension over $\mc C$
equal to $\deg(H)$.
%
Therefore, $H$ satisfies condition (c) of Corollary \ref{20140407:prop1}
over $\widetilde{\mc F}$.
The same argument as in the proof of Proposition \ref{20140407:prop2} 
and Corollary \ref{20140407:prop1} shows that condition (c) of Corollary \ref{20140407:prop1}
implies condition (b),
where $\mc K$ is replaced by $\widetilde{\mc V}$
Hence, we obtain that $(\ker H)^\perp\subset \im H^*$ over $\widetilde{\mc V}$.
We don't assume in this argument that $H$ is skewadjoint,
but if it is, we have $(\ker H)^\perp= \im H$
since, for skewadjoint operators, the opposite inclusion always holds.
To conclude, we clearly have $C(H)=\{\tint F\cdot u\,|\,F\in\ker H\}$,
so that $\ker H=\delta C(H)$.

\end{proof}

\subsection{Hamiltonian equations and integrability}\label{sec:2.2}

Recall that the \emph{Hamiltonian partial differential equation}
for the Poisson structure $H$ on $\mc V$
and the Hamiltonian functional $\tint h\in\mc V/\partial\mc V$ 
is the following evolution equation in ${\bf u}=(u_i)_{i\in I}$:
\begin{equation}\label{20130222:eq10}
\frac{d{\bf u}}{dt}
=
H(\partial)\delta h
\,.
\end{equation}
An \emph{integral of motion} for the Hamiltonian equation \eqref{20130222:eq10}
is a local functional $\tint f\in\mc V/\partial\mc V$ such that $\{\tint h,\tint f\}_H=0$.
Equation \eqref{20130222:eq10} is said to be \emph{integrable}
if there is an infinite sequence of linearly independent integrals of motion 
$\tint h_0=\tint h,\tint h_1,\tint h_2,\dots$
in involution: $\{\tint h_m,\tint h_n\}_H=0$ for all $m,n\in\mb Z_+$.
In this case, we have an \emph{integrable hierarchy} of Hamiltonian PDE:
$$
\frac{d{\bf u}}{dt_n}=H(\partial)\delta h_n
\,\,,\,\,\,\,
n\in\mb Z_+
\,.
$$

\subsection{Lenard-Magri scheme}\label{sec:2.3}

One of the main tools to prove integrability is the following well-known result.
\begin{proposition}
Let $\mc V$ be an algebra of differential functions,
and let $H_0,H_1\in\Mat_{\ell\times\ell}\mc V[\partial]$ be skew-adjoint matrix differential operators.
Denote by $\{\cdot\,,\,\cdot\}_{0},\,\{\cdot\,,\,\cdot\}_{1}$
be the corresponding brackets on $\mc V/\partial\mc V$
(cf. \eqref{20130222:eq7}).
Let 
$\{\tint h_n\}_{n\in\mb Z_+}\subset\mc V/\partial\mc V$ 
be an infinite sequence satisfying the following recurrence relation
\begin{equation}\label{20140404:eq3}
H_1(\partial)\delta h_{n-1}
=
H_0(\partial)\delta h_n
\,,
\end{equation}
for every $n\in\mb Z_+$
(where, for $n=0$, we let $\tint h_{-1}=0$).
Then 
\begin{enumerate}[(a)]
\item
The elements $\tint h_n$ are integrals of motion in involution
with respect to both brackets $\{\cdot\,,\,\cdot\}_{0,1}$:
\begin{equation}\label{20140404:eq6}
\big\{\tint h_m,\tint h_n\big\}_{0,1}=0
\,\,\text{ for all }\,\, m,n\in\mb Z_+
\,.
\end{equation}
\item
Let 
$\{\tint g_n\}_{n=0}^N\subset\mc V/\partial\mc V$
be any other sequence satisfying the same recurrence relations \eqref{20140404:eq3}
for $n=0,\dots,N$.
Then
\begin{equation}\label{20140404:eq6b}
\big\{\tint h_m,\tint g_n\big\}_{0,1}=0
\,\,\text{ for all }\,\, m,n\in\mb Z_+
\,.
\end{equation}
\item
If $H_0$ is a Poisson structure over $\mc V$,
then the evolutionary vector fields $X_{P_n}$, where $P_n=H_0\delta h_n$,
commute.
\end{enumerate}
\end{proposition}
\begin{proof}
A proof can be found in \cite{Mag78} for part (a) and in \cite{BDSK09} for part (b) 
(or both in \cite{DSKT13}).
For a proof of (c), see e.g. \cite[Prop.1.2.4(c)]{BDSK09}.
\end{proof}

The following result provides a way to solve, recursively, the Lenard-Magri relations \eqref{20140404:eq3},
and therefore to prove integrability.
\begin{theorem}[\cite{BDSK09,DSKT13}]\label{20140403:thm}
Let $\mc V$ be an algebra of differential functions (which is assumed to be a domain),
and let $\widetilde{\mc V}$ be a normal extension of $\mc V$.
Let $H_0,H_1\in\Mat_{\ell\times\ell}\mc V[\partial]$ be two compatible Poisson structures 
on $\mc V$.
Assume that $H_0$ is strongly skew-adjoint over $\mc V$.
Let $\tint h_0\in C(H_0)$.
Then: 
\begin{enumerate}[(a)]
\item
There exists $\tint h_1\in\widetilde{\mc V}/\partial\widetilde{\mc V}$ 
such that $\delta h_1\in\mc V^\ell$ and
\begin{equation}\label{20140404:eq2}
H_1(\partial) \delta h_0
=
H_0(\partial) \delta h_1
\,,
\end{equation}
if and only if
\begin{equation}\label{20140404:eq1}
\big\{\tint h_0,C(H_0)\big\}_1=0
\,.
\end{equation}
\item
Suppose, inductively, that 
$\tint h_0\in C(H_0),\,\tint h_1,\dots,\tint h_N\in\widetilde{\mc V}/\partial\widetilde{\mc V}$,
with $\delta h_i\in\mc V^\ell$,
satisfy the recurrence relations \eqref{20140404:eq3}
for every $n=0,\dots,N$ (where we denote $\tint h_{-1}=0$).
Then
\begin{equation}\label{20140404:eq4}
\big\{\tint h_N,C(H_0)\big\}_1
\subset C(H_0)
\,.
\end{equation}
\item
Suppose, in addition, that
\begin{equation}\label{20140404:eq4b}
\big\{\tint h_N,C(H_0)\big\}_1=0
\,.
\end{equation}
Then, there exists $\tint h_{N+1}\in\widetilde{\mc V}/\partial\widetilde{\mc V}$
such that $\delta h_{N+1}\in\mc V^\ell$ and
\begin{equation}\label{20140404:eq5}
H_1(\partial) \delta h_N
=
H_0(\partial) \delta h_{N+1}
\,.
\end{equation}
\item
If \eqref{20140404:eq2} holds
and 
\begin{equation}\label{20140409:eq1}
\{\delta h_0,\delta h_1\}^\perp\subset \im H_0
\,,
\end{equation}
then the Lenard-Magri recursive equation \eqref{20140404:eq3}
has solution $\tint h_n\in\widetilde{\mc V}/\partial\widetilde{\mc V}$,
with $\delta h_n\in\mc V^\ell$,
for each $n\in\mb Z_+$.
\end{enumerate}
\end{theorem}
\begin{proof}
Condition \eqref{20140404:eq1}
is equivalent to saying that $H_1(\partial) \delta h_0$
is orthogonal to $\frac{\delta}{\delta u}C(H_0)$
with respect to the form $(F,P)\mapsto\tint F\cdot P$.
Hence, by the strong skew-adjointness of $H_0$,
this is in turn equivalent to $H_1(\partial) \delta h_0 \in\im(H_0)$,
proving (a).
Part (b) is an immediate consequence of the compatibility of $H_0$ and $H_1$
and the assumption that $\tint h_0\in C(H_0)$.
(More details can be found in \cite{DSKT13}.)
The proof of (c) is the same as for part (a).
Finally, claim (d) is \cite[Cor.2.12]{BDSK09}.
\end{proof}
The following special case will be used to prove integrability in the present paper.
\begin{corollary}\label{20140409:cor}
Let $\mc V$ be an algebra of differential functions (which is assumed to be a domain),
and let $\widetilde{\mc V}$ be a normal extension of $\mc V$.
Let $H_0,H_1\in\Mat_{\ell\times\ell}\mc V[\partial]$ be two compatible Poisson structures 
on $\mc V$.
Assume that $H_0$ is strongly skew-adjoint.
Suppose that 
\begin{equation}\label{20140409:eq2}
\dim_{\mc C}C(H_0)=2
\,\,\text{ and }\,\,
C(H_0)\cap C(H_1)\neq0
\,.
\end{equation}
Let $\{\tint h_0,\tint h_1\}$ be a basis of $C(H_0)$ over $\mc C$,
where $\tint h_0\in C(H_0)\cap C(H_1)$.
Then, the Lenard-Magri recursive equation \eqref{20140404:eq3}
has solution $\tint h_n\in\widetilde{\mc V}/\partial\widetilde{\mc V}$,
with $\delta h_n\in\mc V^\ell$,
for each $n\in\mb Z_+$.
\end{corollary}
\begin{proof}
Equation \eqref{20140404:eq2} trivially holds,
and condition \eqref{20140409:eq1} follows by the assumption that $H_0$
is strongly skew-adjoint.
Hence, the claim follows form Theorem \ref{20140403:thm}(d).
\end{proof}
\begin{remark}\label{20140409:rem}
Corollary \ref{20140409:cor} still holds if conditions \eqref{20140409:eq2}
are replaced by the assumption that $C(H_0)\cap C(H_1)$
has codimension at most $1$ in $C(H_0)$.
\end{remark}

\section{The $6$-parameter family of compatible Poisson structures}\label{sec:3}

Consider the Lie conformal algebra $L=\mb F[\partial]u\oplus\mb F[\partial]v$
over a field $\mb F$, with the $\lambda$-bracket
$$
\{u_\lambda u\}=(\partial+2\lambda)u
\,\,,\,\,\,\,
\{u_\lambda v\}=(\partial+\lambda)v
\,\,,\,\,\,\,
\{v_\lambda v\}=0
\,.
$$
(It is the semidirect product of the Virasoro conformal algebra and its representation
on currents.)
The space of two-cocycles on $L$ with values in $\mb F$ is $5$-dimensional,
and the generic two-cocycle $\psi_\lambda:\,L\times L\to\mb F[\lambda]$ looks as follows
$$
\begin{array}{l}
\displaystyle{
\vphantom{\Big(}
\psi_\lambda(u,u)=\alpha\lambda+c\lambda^3
\,\,,\,\,\,\,
\psi_\lambda(u,v)=\beta\lambda-\gamma\lambda^2
\,,} \\
\displaystyle{
\vphantom{\Big(}
\psi_\lambda(v,u)=\beta\lambda+\gamma\lambda^2
,\,
\psi_\lambda(v,v)=\epsilon\lambda
\,,}
\end{array}
$$
for $\alpha,\beta,\gamma,\epsilon,c\in\mb F$.
As a result, we get a central extension $\hat{L}$ of $L$ corresponding to this two-cocycle.
The corresponding $\lambda$-bracket on $\hat{L}$
give rise to the following $6$-parameter family of 
(pairwise compatible) Poisson structures on 
$\mb F[u,v,u',v',\dots]$:
\begin{equation}\label{20140404:eq7}
G^{(c,a,\gamma)}_{(\alpha,\beta,\epsilon)}(\partial)
= 
\left(
\begin{array}{cc} 
 a(u'+2u\partial)+\alpha \partial + c\partial^3 & a\, v \partial+\beta \partial + \gamma \partial^2 \\
 a \partial\circ v + \beta \partial -\gamma \partial^2 & \epsilon \partial
\end{array}\right)
\,.
\end{equation}
Throughout the rest of the paper, 
we will denote by $\mc V$ an algebra of differential functions
extension of $\mb F[u,v,u',v',\dots]$,
with subfields of constants and quasiconstant $\mc C=\mc F=\mb F$, unless otherwise specified.
We also assume that $\mc V$ contains all the expressions that will appear
in the text.
For example, if we encounter an element of the form $\frac{u}{v}$,
we assume that $\mc V$ is an extension of $\mb F[u,v^{\pm1},u',v',\dots]$.

\subsection{The case $a\neq0$}

For $a\neq0$ in \eqref{20140404:eq7}, after rescaling we can set $a=1$.
Moreover, in this case after replacing $u$ by $u-\frac{\alpha}{2}$ and $v$ by $v-\beta$,
we may assume that $\alpha=\beta=0$.
Hence, without loss of generality,
we get the following $3$-parameter family of Poisson structures:
\begin{equation}\label{20140404:eq8}
H_{(c,\gamma,\epsilon)}(\partial)
= 
\left(
\begin{array}{cc} 
 u'+2u\partial + c\partial^3 & v \partial + \gamma \partial^2 \\
 \partial\circ v -\gamma \partial^2 & \epsilon \partial
\end{array}\right)
\,.
\end{equation}
We introduce the following important quantities:
\begin{equation}\label{20140408:eq4}
p=\epsilon c+\gamma^2\in\mb F
\,\,\text{ and }\,\,
Q=\epsilon u-\frac12v^2-\gamma v'\in\mc V
\,.
\end{equation}
\begin{lemma}\label{20140404:lema}
We have
\begin{equation}\label{20140404:eq8b}
\left(
\begin{array}{cc} 
\epsilon & -v-\gamma\partial \\
0 & 1
\end{array}\right)
H_{(c,\gamma,\epsilon)}(\partial)
= 
\left(
\begin{array}{cc} 
Q'+2Q\partial+p\partial^3 & 0 \\
\partial\circ v -\gamma \partial^2 & \epsilon \partial
\end{array}\right)
\,.
\end{equation}
The matrix $H_{(c,\gamma,\epsilon)}(\partial)$ is always non-degenerate, 
and its Dieudonn\`e determinant is
\begin{equation}\label{20140408:eq3}
\det(H_{(c,\gamma,\epsilon)})
=
\left\{\begin{array}{l}
p\xi^4\,\,\text{ for }\,\, p\neq0\,, \\
2Q\xi^2\,\,\text{ for }\,\, p=0\,.
\end{array}\right.
\end{equation}
\end{lemma}
\begin{proof}
Equation \eqref{20140404:eq8b} is straightforward to check.
For $\epsilon=0$, the matrix $H_{(c,\gamma,\epsilon)}$ becomes lower triangular after a
permutation of the rows, and in this case the determinant is easily computed.
For $\epsilon\neq0$, \eqref{20140408:eq3} follows immediately by \eqref{20140408:eq2b}.
\end{proof}

\begin{lemma}\label{20140404:lemb}
The kernel of the operator $H_{(c,\gamma,\epsilon)}(\partial)$ in \eqref{20140404:eq8}
is, depending on the values of the parameters $c,\gamma,\epsilon\in\mb F$, as follows:
\begin{enumerate}[(i)]
\item
If $p\neq0$, then
$$
\ker(H_{(c,\gamma,\epsilon)})
=
\mb F\delta v
\,.
$$
\item
If $\epsilon=0$, $\gamma=0$, 
then, in $\mc V\supset R_2[v^{-1}]$
with field of constants $\mb F$, we have
$$
\ker(H_{(c,\gamma,\epsilon)})
=
\mb F
\delta v
\oplus
\mb F
\delta \big(\frac{u}{v}-\frac{c}{2}\frac{(v')^2}{v^3}\big)
\,.
$$
\item
For $\epsilon\neq0$, $p=0$,
then, in $\mc V\supset R_2[Q^{-\frac12}]$
with field of constants $\mb F$, we have
$$
\ker(H_{(c,\gamma,\epsilon)})
=
\mb F\delta v
\oplus
\mb F
\delta (Q^{\frac12})
\,.
$$
\end{enumerate}
\end{lemma}
\begin{proof}
We consider first the case when $\epsilon=0$, $\gamma\neq0$.
In this case, $\Big(\begin{array}{l} f \\ g \end{array}\Big)$ lies in $\ker(H_{(c,\gamma,\epsilon)})$
implies
\begin{equation}\label{20140404:eq11}
f''=\frac1{\gamma}(vf)'\,.
\end{equation}
By a differential order consideration, it is immediate to see that, in this case, $f$ must be a quasiconstant,
and then, applying $\frac{\partial}{\partial v'}$ to both sides of \eqref{20140404:eq11},
we get that $f=0$.
Then, the equation
$H_{(c,\gamma,\epsilon)}(\partial)\Big(\begin{array}{l} 0 \\ g \end{array}\Big)=0$
gives
\begin{equation}\label{20140404:eq12}
g''=-\frac1{\gamma} vg'
\,.
\end{equation}
Again, by a differential order consideration we get that $g$ is quasiconstant,
and then, applying $\frac{\partial}{\partial v}$ to both sides of \eqref{20140404:eq12},
we conclude that $g'=0$.
Therefore, in this case, we have $f=0$, $g\in\mb F$,
which means that 
$\Big(\begin{array}{l} f \\ g \end{array}\Big)=$ const. $\delta v$.

Next, we consider the case when $\epsilon=\gamma=0$.
In this case, 
if $\Big(\begin{array}{l} f \\ g \end{array}\Big)$ lies in $\ker(H_{(c,\gamma,\epsilon)})$,
we get $(vf)^\prime=0$, which means
$f=\frac{\alpha}{v}$
for $\alpha\in\mb F$.
Then, the equation 
$H_{(c,\gamma,\epsilon)}(\partial)\Big(\begin{array}{l} \frac{\alpha}{v} \\ g \end{array}\Big)=0$
gives
$$
\begin{array}{l}
\displaystyle{
\vphantom{\Big(}
g'
=
-\alpha\frac{u'}{v^2}+2\alpha\frac{uv'}{v^3}-\alpha c\frac{1}{v}\Big(\frac{1}{v}\Big)'''
} \\
\displaystyle{
\vphantom{\Big(}
=
\alpha\Big(
-\frac{u}{v^2}
- c \frac{1}{v}\Big(\frac{1}{v}\Big)^{\prime\prime}
+\frac12 c \Big(\Big(\frac{1}{v}\Big)^\prime\Big)^2
\Big)^\prime
} \\
\displaystyle{
\vphantom{\Big(}
=
\alpha\Big(\frac{\delta}{\delta v}\Big(
\frac{u}{v}
- \frac{c}{2} \frac{(v')^2}{v^3}
\Big)\Big)^\prime
\,.}
\end{array}
$$
Hence, in this case, we have 
$$
\Big(\begin{array}{l} f \\ g \end{array}\Big)=
\alpha\delta\Big(
\frac{u}{v}
- \frac{c}{2} \frac{(v')^2}{v^3}
\Big)
+
\beta\delta v
\,,
$$
for some $\alpha,\beta\in\mb F$.
This proves part (ii).

For $\epsilon\neq0$, the matrix 
$\left(
\begin{array}{cc} 
\epsilon & -v-\gamma\partial \\
0 & 1
\end{array}\right)$
is invertible.
Hence, by equation \eqref{20140404:eq8b},
$\Big(\begin{array}{l} f \\ g \end{array}\Big)$ lies in $\ker(H_{(c,\gamma,\epsilon)})$
if and only if
\begin{equation}\label{20140404:eq10}
\left\{\begin{array}{l}
\displaystyle{
\vphantom{\Big(}
Q'f+2Qf'+pf'''=0
\,,} \\
\displaystyle{
\vphantom{\Big(}
(vf)^\prime-\gamma f''+\epsilon g'=0
\,.}
\end{array}\right.
\end{equation}
For $p\neq0$,
we can use a simple differential order consideration 
on the first equation of \eqref{20140404:eq10}
to prove that $f$ is quasiconstant,
and then, looking at the coefficient of $vv'$, 
we conclude that $f=0$.
Therefore, in this case, we have $f=0$, $g\in\mb F$.
This concludes the proof of part (i).

Finally, for $p=0$ the first equation in \eqref{20140404:eq10} becomes
\begin{equation}\label{20140404:eq15}
2 Q f' + Q'f = 0
\,,
\end{equation}
which has solution 
\begin{equation}\label{20140404:eq16}
f
=
\frac12\alpha\epsilon\, Q^{-\frac12}
=
\alpha\frac{\delta}{\delta u}Q^{\frac12}
\,,
\end{equation}
for $\alpha\in\mb F$.
Substituting in the second equation of \eqref{20140404:eq10}, we get
\begin{equation}\label{20140404:eq17}
g
=
\frac12\gamma\alpha (Q^{-\frac12})'
-\frac12\alpha v Q^{-\frac12}
+
\beta
=
\alpha\frac{\delta}{\delta v}Q^{\frac12}+\beta\frac{\delta}{\delta v}v
\,,
\end{equation}
for $\beta\in\mb F$.
In conclusion, 
$\Big(\begin{array}{l} f \\ g \end{array}\Big)=\alpha\delta (Q^{\frac12})+\beta\delta v$,
proving (iii).
\end{proof}
\begin{corollary}\label{20140408:cor}
The operator $H_{(c,\gamma,\epsilon)}$ in \eqref{20140404:eq8}
is strongly skew-adjoint over $\mc V$
if and only if $p=0$.
\end{corollary}
\begin{proof}
%
%
First, for $p\neq0$ we have, by Lemma \ref{20140404:lemb}(i),
$\ker H_{(c,\gamma,\epsilon)}=\mb F\left(\begin{array}{l} 0\\1 \end{array}\right)$.
Hence, the orthogonal complement to $\ker H_{(c,\gamma,\epsilon)}$ consists of all vectors
of the form $\left(\begin{array}{l} f\\g' \end{array}\right)$,
with $f,g\in\mc V$.
We need to prove that not all these elements lie in the image of $H_{(c,\gamma,\epsilon)}$.
Suppose, by contradiction, that
$\left(\begin{array}{l} f\\0 \end{array}\right)\in\im H_{(c,\gamma,\epsilon)}$
for every $f\in\mc V$.
Then, by Lemma \ref{20140404:lemb},
we have $f\in\im(Q'+2Q\partial+p\partial^3)$,
which is impossible, since the operator $Q'+2Q\partial+p\partial^3$ is not surjective on $\mc V$.

In the case $\epsilon=\gamma=0$, 
$H_{(c,\gamma,\epsilon)}$ is strongly skewadjoint due to \cite[Prop.5.1(c)]{DSKT13}.
Finally, consider the case $p=0,\epsilon\neq0$.
By Lemma \ref{20140404:lemb}(iii), the kernel of $H_{(c,\gamma,\epsilon)}$ 
is spanned over $\mb F$ by
$\left(\begin{array}{l} 0\\1 \end{array}\right)$
and 
$\left(\begin{array}{c} \frac{\epsilon}{2} Q^{-\frac12} \\ 
-\frac12 vQ^{-\frac12}+\frac{\gamma}{2}\partial Q^{-\frac12} \end{array}\right)$.
Hence, the orthogonal complement to $\ker H_{(c,\gamma,\epsilon)}$ consists of
elements of the form
$\left(\begin{array}{l} f\\g' \end{array}\right)$,
where $f,g\in\mc V$ are such that
$$
\frac12\int Q^{-\frac12}
\Big(
\epsilon f-vg'-\gamma g''
\Big)=0
\,.
$$
In other words (by construction, $Q^{\frac12}$ is an invertible element of $\mc V$),
$$
\epsilon f
=
vg'+\gamma g''
+2 Q^{\frac12}\partial(Q^{\frac12} h)
\,,
$$
for some $h\in\mc V$.
Hence, we need to prove that, for every $g,h\in\mc V$, we have
$$
\left(\begin{array}{l} f\\g' \end{array}\right)
=
\left(\begin{array}{l} 
\frac1\epsilon\Big(
vg'+\gamma g''+2 Q^{\frac12}\partial(Q^{\frac12} h)
\Big) \\
g'
\end{array}\right)
\,\,\in
\im H_{(c,\gamma,\epsilon)}
\,.
$$
By Lemma \ref{20140404:lema}, this is the same as proving that
$$
\left(
\begin{array}{cc} 
\epsilon & -v-\gamma\partial \\
0 & 1
\end{array}\right)
\left(\begin{array}{l} f\\g' \end{array}\right)
=
\left(\begin{array}{l} 
2 Q^{\frac12}\partial(Q^{\frac12} h) \\
g'
\end{array}\right)
\,\,\in
\im
\left(
\begin{array}{cc} 
2 Q^{\frac12}\partial\circ Q^{\frac12} & 0 \\
\partial\circ v -\gamma \partial^2 & \epsilon \partial
\end{array}\right)
\,,
$$
which is easily checked.

Also, $\ker H_{(c,\gamma,\epsilon)}\subset\delta(\mc V)$
by Lemma \ref{20140404:lemb}.
Hence $H_{(c,\gamma,\epsilon)}$ is strongly skew-adjoint.
\end{proof}
\begin{remark}\label{20140505:rem}
The fact that $H_{(c,\gamma,\epsilon)}$ is strongly skew-adjoint over 
the differential field $\mc K=Frac(\mc V)$
when $p=0$ is immediate by Corollary \ref{20140407:prop1} and Lemma \ref{20140404:lemb}.
\end{remark}

\subsection{The case $a=0$}

When $a=0$ in \eqref{20140404:eq7}, we have the following $5$-parameter family of 
Poisson structures with constant coefficients:
\begin{equation}\label{20140404:eq9}
H^{(c,\alpha,\beta,\gamma,\epsilon)}(\partial)
= 
\left(
\begin{array}{cc} 
\alpha \partial + c\partial^3 & \beta \partial + \gamma \partial^2 \\
\beta \partial -\gamma \partial^2 & \epsilon \partial
\end{array}\right)
\,.
\end{equation}
\begin{corollary}\label{20140408:cor2}
The operator $H^{(c,\alpha,\beta,\gamma,\epsilon)}(\partial)$ in \eqref{20140404:eq9}
is non-degenerate if and only if either $p\neq0$ or $q:=\alpha\epsilon-\beta^2\neq0$,
and in this case it is strongly skew-adjoint
(for some differential field $\mc F$ of quasiconstants
with the field of constants $\mc C$).
\end{corollary}
\begin{proof}
It follows from Proposition \ref{20140407:prop}.
\end{proof}
As we did in the previous section,
we compute the kernel of the matrix differential operator  $H^{(c,\alpha,\beta,\gamma,\epsilon)}$.
\begin{lemma}\label{20140410:lem}
Let, as before, $p=c\epsilon+\gamma^2$ and $q=\alpha\epsilon-\beta^2$.
The kernel of the operator $H^{(c,\alpha,\beta,\gamma,\epsilon)}(\partial)$ in \eqref{20140404:eq9}
is, depending on the values of the parameters $c,\alpha,\beta,\gamma,\epsilon\in\mb F$, as follows.
If $p=0$, $q\neq0$, then 
$$
\ker H^{(c,\alpha,\beta,\gamma,\epsilon)}(\partial)
=
\Span{}_{\mb F}
\{\delta v,\delta u\}
\,.
$$
If $p\neq0$, then $\ker H^{(c,\alpha,\beta,\gamma,\epsilon)}(\partial)$ is $4$-dimensional 
over $\mc C$
with basis $\delta f_i$, $i=1,2,3,4$, where $f_1=v$, $f_2=u$,
and the elements $f_3$ and $f_4$ as follows:
\begin{enumerate}[(i)]
\item
if $q\epsilon\neq0$, then,
taking $\mc C=\mb F[\sqrt{\frac{q}{p}}]$
and $\mc F=\mc C(\cos\sqrt{\frac{q}{p}}x,\sin\sqrt{\frac{q}{p}}x)$,
$$
\begin{array}{l}
\displaystyle{
f_3
=
\Big(
\cos{\sqrt{\frac{q}{p}}x}
\Big)
u
-\Big(
\frac{\beta}{\epsilon}
\cos{\sqrt{\frac{q}{p}}x}
+\frac{\gamma}{\epsilon}
\sqrt{\frac{q}{p}}
\sin{\sqrt{\frac{q}{p}}x}
\Big)v
\,,} \\
\displaystyle{
f_4
=
\Big(
\sin{\sqrt{\frac{q}{p}}x}
\Big)u 
+\Big(
\frac{\gamma}{\epsilon}
\sqrt{\frac{q}{p}}
\cos{\sqrt{\frac{q}{p}}x}
-\frac{\beta}{\epsilon}
\sin{\sqrt{\frac{q}{p}}x}
\Big)v
\Bigg\}
\,;} \\
\end{array}
$$
\item
if $\epsilon=0,q\neq0$, then,
taking $\mc C=\mb F$
and $\mc F=\mb F(e^{\pm\frac\beta\gamma x})$,
$$
f_3
=
2
e^{\frac{\beta}{\gamma}x}u
+\Big(
\frac{\alpha}{\beta}
+
\frac{c\beta}{\gamma^2}
\Big)
e^{\frac{\beta}{\gamma}x }v
\,\,,\,\,\,\,
f_4
=
e^{-\frac{\beta}{\gamma}x} v
\,;
$$
\item
if $\epsilon\neq0$, $q=0$, then,
taking $\mc C=\mb F$
and $\mc F=\mb F(x)$,
$$
f_3
=
-x^2 u+\Big(\frac{\beta}{\epsilon}x^2-2\frac{\gamma}{\epsilon}x\Big)v
\,\,,\,\,\,\,
f_4
=
-xu+\frac{\beta}{\epsilon}xv
\,;
$$
\item
if $\epsilon=q=0$, then,
taking $\mc C=\mb F$
and $\mc F=\mb F(x)$,
$$
f_3
=
-xu+\frac{\alpha}{2\gamma}x^2v
\,\,,\,\,\,\,
f_4
=
xv
\,.
$$
\end{enumerate}
\end{lemma}
\begin{proof}
Straightforward.
\end{proof}

\section{Bi-Hamiltonian hierarchies with $H_0$ of type \eqref{20140404:eq8}}
\label{sec:4}

In this and the next section
we discuss the applicability of the Lenard-Magri scheme of integrability
for a bi-Poisson structure $(H_0,H_1)$
within the family \eqref{20140404:eq7}.
In order to use Theorem \ref{20140403:thm}
we shall assume that $H_0$ is strongly skew-adjoint.
This leads to the following two cases (cf. Corollaries \ref{20140408:cor} and \ref{20140408:cor2}),
where, as before,
$p=c\epsilon+\gamma^2$
and
$q=\alpha\epsilon-\beta^2$:
\begin{enumerate}[A)]
\item
the case when
$H_0=H_{(c,\gamma,\epsilon)}$, for $c,\gamma,\epsilon\in\mb F$
such that $p=0$,
which will be discussed in the present section;
\item
the case when
$H_0=H^{(c,\alpha,\beta,\gamma,\epsilon)}$, for $c,\alpha,\beta,\gamma,\epsilon\in\mb F$
with either $p\neq0$, or $p=0$ and $q\neq0$,
which will be discussed in Section \ref{sec:5}.
\end{enumerate}

Recall that the Lenard-Magri scheme associated to a bi-Poisson structure $(H_0,H_1)$
only depends on the flag $\mb FH_0\subset\mb FH_0+\mb FH_1$,
see e.g. \cite{BDSK09}.
Hence, in case (a), adding a constant multiple of $H_0$ to $H_1$,
we may assume that $H_1$ is of type \eqref{20140404:eq9}.
Therefore, we consider the following bi-Poisson structures
$$
H_0=H_{(c,\gamma,\epsilon)}
\,\,,\,\,\,\,
H_1=H^{(c_1,\alpha_1,\beta_1,\gamma_1,\epsilon_1)}
\,,
$$
where $c,\gamma,\epsilon,c_1,\alpha_1,\beta_1,\gamma_1,\epsilon_1\in\mb F$
and $p=0$.
In this section
\begin{equation}\label{20140416:eq3a}
\mc V=\mb F[u,v^{\pm1},u',v',u'',v'',\dots]
\,\,\text{ if }\,\, 
\epsilon=0
\,,
\end{equation}
and
\begin{equation}\label{20140416:eq3b}
\mc V=\mb F[u,v,u',v',u'',v'',\dots,Q^{-\frac12}]
\text{ if } \epsilon\neq0
\,.
\end{equation}

Let $\tint h_0=\tint v$,
and let $\tint h_1\in C(H_0)$ be equal to 
$\tint \big(\frac{u}{v}-\frac{c}{2}\frac{(v')^2}{v^3}\big)$ or $\tint Q^{\frac12}$,
depending on whether $\epsilon=0$ or $\epsilon\neq0$.
By Lemma \ref{20140404:lemb}, $\{\tint h_0,\tint h_1\}$
is a basis of $C(H_0)$.
Moreover, $\tint h_0=\tint v\in C(H_1)$.
Therefore, by Corollary \ref{20140409:cor},
there exists an infinite sequence $\{\tint h_n\}_{n\in\mb Z_+}$
satisfying the Lenard-Magri recursive conditions \eqref{20140404:eq3}
for each $n\in\mb Z_+$.
We thus get an integrable hierarchy of bi-Hamiltonian equations
\begin{equation}\label{20140409:eq3}
\frac{d{\bf u}}{dt_n}=H_0(\partial)\delta h_n
\,,
\end{equation}
provided that the integrals of motion $\tint h_n$ are linearly independent,
which follows from
\begin{proposition}\label{20140416:prop}
If $H_1\neq0$, then the elements $\delta h_n=\left(\begin{array}{l}f_n \\ g_n\end{array}\right)\in\mc V^2$, 
$n\in\mb Z_+$, are linearly independent.
\end{proposition}
\begin{proof}
First, consider the case $\epsilon=0$.
In this case, the Lenard-Magri recursive formula \eqref{20140404:eq3} reads
\begin{equation}\label{20140416:eq1}
\begin{array}{l}
\displaystyle{
\vphantom{\Big(}
(u'+2u\partial+c\partial^3)f_n+v\partial g_n
=
(\alpha_1\partial+c_1\partial^3)f_{n-1}+(\beta_1\partial+\gamma_1\partial^2)g_{n-1}
\,,} \\
\displaystyle{
\vphantom{\Big(}
\partial(vf_n)
=
(\beta_1\partial-\gamma_1\partial^2)f_{n-1}
+\epsilon_1\partial g_{n-1}
\,.}
\end{array}
\end{equation}
The first two elements $\delta h_0$ and $\delta h_1$
form a basis of $\ker(H_0)$,
and they are
\begin{equation}\label{20140416:eq2}
\left(\begin{array}{l} f_0\\g_0 \end{array}\right)
=
\left(\begin{array}{l} 0\\1 \end{array}\right)
\,\,\text{ and }\,\,
\left(\begin{array}{l} f_1\\g_1 \end{array}\right)
=
\left(\begin{array}{l} v^{-1}\\-v^{-2}u+cv^{-3}v''-\frac32cv^{-4}(v')^2 \end{array}\right)
\,.
\end{equation}
Hence, they lie in $\mc V^2$, where $\mc V$ is as in \eqref{20140416:eq3a}.
In particular, by Corollary \ref{20140408:cor}, $H_0$ is strongly skewadjoint over $\mc V$,
and, by Corollary \ref{20140409:cor},
the recursive equation \eqref{20140416:eq1} has solution
$f_n,g_n\in\mc V$ for all $n\in\mb Z_+$.
The algebra $\mc V$ in \eqref{20140416:eq3a}
is $\mb Z$-graded by the total polynomial degree, where $u,v,u',v',\dots$ have all degree $+1$,
while $v^{-1}$ has degree $-1$.
We thus have
$\mc V=\bigoplus_{k\in\mb Z}\mc V[k]$,
where $\mc V[k]$ is the subspace of homogeneous elements of degree $k\in\mb Z$.
With respect to this grading, $\deg(f_1)=0$, and $\deg(f_1)=\deg(g_1)=-1$.
Furthermore, it follows by the recursive equations \eqref{20140416:eq1}
and an easy induction on $n$, that
$\deg(f_n),\deg(g_n)\leq-n$ for all $n\geq1$.
Let $f_n^0,g_n^0\in\mc V[-n]$ be the homogeneous components in $f_n$ and $g_n$
of degree $-n$.
Recalling \eqref{20140416:eq1}, they satisfy the recursive equations:
\begin{equation}\label{20140416:eq4}
\left(\begin{array}{l}
(u'+2u\partial)f_n^0+v\partial g_n^0
\\
\partial(vf_n^0)
\end{array}\right)
=
H_1(\partial)
\left(\begin{array}{l}f_{n-1}^0 \\ g_{n-1}^0\end{array}\right)
\,.
\end{equation}
We prove, by induction on $n\geq2$, 
that $\left(\begin{array}{l}f_n^0 \\ g_n^0\end{array}\right)\neq0$.
Since, by assumption, $H_1$ has constant coefficients and it is non-degenerate,
it follows that $\ker H_1\subset\mb F^2$.
Therefore, 
since $\left(\begin{array}{l}f_{n-1}^0 \\ g_{n-1}^0\end{array}\right)\in\mc V[-n+1]^2\backslash\{0\}$
and $\mc V[-n+1]^2\cap\mb F^2=0$,
it follows that $\left(\begin{array}{l}f_{n-1}^0 \\ g_{n-1}^0\end{array}\right)\not\in\ker H_1$.
In other words, the RHS of \eqref{20140416:eq4} is not zero,
which implies that $\left(\begin{array}{l}f_n^0 \\ g_n^0\end{array}\right)\neq0$.
In conclusion, $\max\{\deg(f_n),\deg(g_n)\}=-n$,
which implies that the $\delta h_n$'s are linearly independent.

Next, we consider the case $\epsilon\neq0$,
and we let $\mc V$ be as in \eqref{20140416:eq3b}.
Note that,
since $u=\frac1\epsilon(Q+\frac12 v^2+\gamma v')$, we have
$$
\mc V=\mb F[u,v,u',v',\dots][Q^{-\frac12}]=\mb F[Q^{\pm\frac12},v,Q',v',Q'',v'',\dots]
\,,
$$
so it can be considered as an algebra of differential functions in the new variables $Q$ and $v$.

We start analyzing the recursive equation \eqref{20140404:eq3}.
Letting, as before, 
$\delta h_n=
\left(\begin{array}{l} f_n \\ g_n \end{array}\right)$,
we get, by equation \eqref{20140404:eq8b},
$$
\begin{array}{l}
\displaystyle{
\left(
\begin{array}{cc} 
Q'+2Q\partial & 0 \\
\partial\circ v -\gamma \partial^2 & \epsilon \partial
\end{array}\right)
\left(\begin{array}{l} f_n \\ g_n \end{array}\right)
} \\
\displaystyle{
=
\left(\begin{array}{cc} 
\epsilon & -v-\gamma\partial \\
0 & 1
\end{array}\right)
\left(\begin{array}{cc} 
\alpha_1 \partial + c_1\partial^3 & \beta_1 \partial + \gamma_1 \partial^2 \\
\beta_1 \partial -\gamma_1 \partial^2 & \epsilon_1 \partial
\end{array}\right)
\left(\begin{array}{l} f_{n-1} \\ g_{n-1} \end{array}\right)
\,,}
\end{array}
$$
which gives the following system of differential equations
\begin{equation}\label{20140424:eq1}
\begin{array}{l}
\displaystyle{
(Q'+2Q\partial)f_n
=
\big(
\epsilon(\alpha_1 \partial + c_1\partial^3)
-(v+\gamma\partial)(\beta_1 \partial -\gamma_1 \partial^2)
\big)
f_{n-1}
} \\
\displaystyle{
\,\,\, \,\,\, \,\,\, \,\,\, \,\,\, \,\,\, \,\,\, \,\,\, \,\,\, \,\,\, \,\,\, \,\,\, 
+
\big(
\epsilon(\beta_1 \partial + \gamma_1 \partial^2)
-(v+\gamma\partial)\epsilon_1 \partial
\big)
g_{n-1}
\,,} \\
\displaystyle{
(vf_n)^\prime-\gamma f_n''+\epsilon g_n'
=
(\beta_1 \partial -\gamma_1 \partial^2)f_{n-1}
+\epsilon_1 \partial g_{n-1}
\,.}
\end{array}
\end{equation}
The second equation in \eqref{20140424:eq1}
allows us to express $g_n$ in terms of $f_n$, $f_{n-1}$, $g_{n-1}$, and their derivatives, 
up to adding a constant:
\begin{equation}\label{20140424:eq2}
g_n
=
\frac{\gamma}{\epsilon} f_n'
-\frac1\epsilon vf_n
-\frac{\gamma_1}{\epsilon} f_{n-1}'
+\frac{\beta_1}{\epsilon} f_{n-1}
+\frac{\epsilon_1}{\epsilon} g_{n-1}
+\text{ const}
\,.
\end{equation}
Using equation \eqref{20140424:eq2} with $n$ replaced by $n-1$,
we can rewrite the first equation in \eqref{20140424:eq1}
as follows
\begin{equation}\label{20140424:eq3}
\begin{array}{l}
\displaystyle{
(Q'+2Q\partial)f_n
=
rf_{n-1}'''+Af_{n-1}'+\frac12 A' f_{n-1}
+\frac{s}{\epsilon}(\gamma_1f_{n-2}'''-\epsilon_1 g_{n-2}'')
} \\
\displaystyle{
+\frac{\epsilon_1}{\epsilon}(-\gamma\beta_1+\gamma_1 v) f_{n-2}''
+\frac{\epsilon_1}{\epsilon}(\epsilon\beta_1-\epsilon_1 v) g_{n-2}'
+\frac{\beta_1}{\epsilon}(\epsilon\beta_1-\epsilon_1 v) f_{n-2}'
\,,}
\end{array}
\end{equation}
where we introduced the following notation:
\begin{equation}\label{20140424:eq3b}
r=\epsilon c_1+2\gamma\gamma_1+c\epsilon_1
\,\,,\,\,\,\,
s=\gamma\epsilon_1-\epsilon\gamma_1
\,\,,\,\,\,\,
A=\epsilon\alpha_1-2\beta_1 v+\frac{\epsilon_1}{\epsilon}v^2+2\frac{s}{\epsilon}v'
\,.
\end{equation}

We define on $\mc V$ a $\frac12\mb Z$ grading letting $\deg(Q^k)=-k$,
and letting all other variables $v,Q',v',Q'',v'',\dots$ have degree $0$.
Hence,
$\mc V=\bigoplus_{k\in\frac12\mb Z}\mc V[k]$, where
$$
\mc V[k]=Q^{-k}\mc V[0]
\,\,,\,\,\,\,
\mc V[0]=\mb F[v,Q',v',Q'',v'',\dots]
\,.
$$
Note that $\mc V[0]$ is a differential subalgebra of $\mc V$.
Moreover, the total derivative $\partial:\,\mc V\to\mc V$ decomposes as
$$
\partial=\partial_0+\partial_1\,,
$$
where $\partial_i:\,\mc V[k]\to\mc V[k+i],\,i=0,1$.
The action of $\partial_0$ and $\partial_1$ on $\mc V[k]$ is given by
$$
\partial_0=Q^{-k}\partial\circ Q^k
\,\,,\,\,\,\,
\partial_1=-kQ^{-1}Q'
\,.
$$
We want to argue, using equations \eqref{20140424:eq2} and \eqref{20140424:eq3},
that the degree of the elements $f_n$ and $g_n$ are strictly increasing in $n$
(thus proving the linear independence of the $\delta h_n$'s).

First, for $n=0$ and $n=1$ we have
$$
f_0=0
\,\,,\,\,\,\,
g_0=1
\,\,,\,\,\,\,
f_1=\frac12 \epsilon Q^{-\frac12}
\,\,,\,\,\,\,
g_1=-\frac14 \gamma Q^{-\frac32} Q' - \frac12 v Q^{-\frac12}
\,.
$$
Hence, $\deg(f_1)=\frac12$ and $\deg(g_1)$ is $1$ for $\gamma\neq0$ and $0$ for $\gamma=0$.

For $r\neq0$
it is easy to prove by induction on $n$ that, for $n\geq2$,
\begin{equation}\label{20140424:eq4}
\deg(f_n)
=
\deg(f_{n-1})+3
\,\,,\,\,\,\,
\deg(g_n)
=
\deg(f_n)+(1-\delta_{\gamma,0})
\,.
\end{equation}

For $r=0$ and $s\neq0$,
one can check, by induction on $n$, that, for $n\geq2$,
\begin{equation}\label{20140501:eq1}
\deg(f_n)
\leq
\left\{\begin{array}{l}
\frac{3n-3}{2}\,\,\text{ for }\,\, n \,\,\text{ even}\\
\frac{3n-2}{2}\,\,\text{ for }\,\, n \,\,\text{ odd}
\end{array}\right.
\,\,\text{ and }\,\,
\deg(g_n)
\leq
\left\{\begin{array}{l}
\frac{3n-1}{2}\,\,\text{ for }\,\, n \,\,\text{ even}\\
\frac{3n}{2}\,\,\text{ for }\,\, n \,\,\text{ odd}
\end{array}\right.
\,.
\end{equation}
Furthermore, 
we have the following recursive equations for the highest degree components 
in $f_n$ and $g_n$
(we denote by $f[k]$ the component in $\mc V[k]$ of $f\in\mc V$).
For $n$ odd:
\begin{equation}\label{20140501:eq2a}
\begin{array}{l}
\displaystyle{
f_n\Big[\frac{3n-2}{2}\Big]
=
-\frac{(3n-8)(3n-6)(3n-4)}{8(3n-3)}\frac{s^2}{\epsilon^2}(Q')^2 Q^{-3} 
f_{n-2}\Big[\frac{3n-8}{2}\Big]
\,,} \\
\displaystyle{
g_n\Big[\frac{3n}{2}\Big]
=
-\frac{(3n-2)}{2} \frac{\gamma}{\epsilon}Q' Q^{-1} f_n\Big[\frac{3n-2}{2}\Big]
\,,}
\end{array}
\end{equation}
and for $n$ even:
\begin{equation}\label{20140501:eq2b}
\begin{array}{l}
\displaystyle{
f_n\Big[\frac{3n-3}{2}\Big]
=
\frac{3n-5}{2(3n-4)}Q^{-1} f_{n-1}\Big[\frac{3n-5}{2}\Big]
} \\
\displaystyle{
-\frac{(3n-9)(3n-7)(3n-5)}{8(3n-4)}\frac{s^2}{\epsilon^2}(Q')^2 Q^{-3} 
f_{n-2}\Big[\frac{3n-9}{2}\Big]
\,,} \\
\displaystyle{
g_n\Big[\frac{3n-1}{2}\Big]
=
-\frac{(3n-3)}{2} \frac{\gamma}{\epsilon}Q' Q^{-1} f_n\Big[\frac{3n-3}{2}\Big]
\,.}
\end{array}
\end{equation}
Since, by assumption, $s\neq0$,
we immediately get from the first recursive formula in \eqref{20140501:eq2a}
that $f_n\Big[\frac{3n-2}{2}\Big]\neq0$ for every odd $n$,
namely $\deg(f_n)=\frac{3n-2}{2}$ for $n$ odd
(which suffices for the proof of integrability).
Actually, by exploiting more the recursions \eqref{20140501:eq2a} and \eqref{20140501:eq2b},
one can  prove that $f_n\Big[\frac{3n-3}{2}\Big]\neq0$ for even $n$ as well.
In conclusion, we get
\begin{equation}\label{20140501:eq3}
\deg(f_n)
=
\left\{\begin{array}{l}
\frac{3n-3}{2}\,\,\text{ for }\,\, n \,\,\text{ even}\\
\frac{3n-2}{2}\,\,\text{ for }\,\, n \,\,\text{ odd}
\end{array}\right.
\,\,\text{ and }\,\,
\deg(g_n)
=
\deg(f_n)+(1-\delta_{\gamma,0})
\,.
\end{equation}

Finally, we consider the case when $s=r=0$.
It is easy to check,
by induction on $n$, that, for $n\geq2$,
\begin{equation}\label{20140501:eq1x}
\deg(f_n)
\leq n-\frac12
\,\,\text{ and }\,\,
\deg(g_n)
\leq n+\frac12
\,.
\end{equation}
Using \eqref{20140501:eq1x},
we can find  recursive equations for the highest degree components of $f_n$ and $g_n$:
\begin{equation}\label{20140501:eq2x}
\begin{array}{l}
\displaystyle{
f_n\Big[n-\frac12\Big]
=
\frac{n-\frac32}{2(n-1)} A Q^{-1} f_{n-1}\Big[n-\frac32\Big]
\,,} \\
\displaystyle{
g_n\Big[n+\frac12\Big]
=
-\Big(n-\frac12\Big)\frac{\gamma}{\epsilon} Q'Q^{-1} f_n\Big[n-\frac12\Big]
\,.}
\end{array}
\end{equation}
From the first recursion \eqref{20140501:eq2x} 
we immediately get
that $f_n\Big[n-\frac12\Big]\neq0$ unless $A=0$,
and using this we obtain that 
equalities in \eqref{20140501:eq1x} hold,
unless $A=0$.
To conclude, we just observe that $s=r=A=0$ implies that $H_1=0$.
\end{proof}

We can compute explicitly 
the first non-trivial equation of the hierarchy
$\frac{d{\bf u}}{dt_2}=P_2=H_1(\partial)\delta h_1$,
and the conserved density $h_2$,
which is a solution of $H_0(\partial)\delta h_2=P_2$,
in both cases, when $\epsilon=0$ or $\epsilon\neq0$.
Since $\int h_0=\int v \in C(H_0)\cap C(H_1)$, we construct the Lenard-Magri scheme 
starting with $\int h_1 \in C(H_0)$ equal to $\int (\frac{u}{v}-\frac{c}{2}\frac{(v')^2}{v^3})$ 
in the case $\epsilon=0$ and $\int Q^{\frac{1}{2}}$ in the case $\epsilon\ne 0$.

{\bf Case A1:} $\epsilon=0$,
and $\mc V=\mb F[u,v^{\pm1},u',v',u'',v''\dots]$.
In this case, the first non-zero equation is:
\begin{equation}\label{turhan:eq1}
\begin{array}{rcl}
\displaystyle{
\frac{du}{dt_2} 
}
&=&
\displaystyle{
c\gamma_1\frac{v^{(iv)}}{v^3}
-9c\gamma_1\frac{v'''v'}{v^4}
-c_1\frac{v'''}{v^2}
+c\beta_1\frac{v'''}{v^3}
-6c\gamma_1\frac{{v''}^2}{v^4}
+42c\gamma_1\frac{v''{v'}^2}{v^5} 
} \\
&&
\displaystyle{
+6c_1\frac{v''v'}{v^3}
-6c\beta_1\frac{v''v'}{v^4}
+2\gamma_1\frac{v''u}{v^3}
-\gamma_1\frac{u''}{v^2}
-30c\gamma_1\frac{{v'}^4}{v^6}
-6c_1\frac{{v'}^3}{v^4}
} \\
&&
\displaystyle{
+6c\beta_1\frac{{v'}^3}{v^5}
-6\gamma_1\frac{{v'}^2{u}}{v^4}
+4\gamma_1\frac{u'v'}{v^3}
-\alpha_1\frac{v'}{v^2}
+2\beta_1\frac{v'u}{v^3}
-\beta_1\frac{u'}{v^2}\,,
} \\
\displaystyle{
\frac{dv}{dt_2} 
}
&=&
\displaystyle{
c\epsilon_1\frac{v'''}{v^3}
-6c\epsilon_1\frac{v''v'}{v^4}
+\gamma_1\frac{v''}{v^2}
+6c\epsilon_1\frac{{v'}^3}{v^5}
-2\gamma_1\frac{{v'}^2}{v^3}
-\beta_1\frac{v'}{v^2}
} \\
&&
\displaystyle{
+2\epsilon_1\frac{v'u}{v^3}
-\epsilon_1\frac{u'}{v^2}
\,,}
\end{array} 
\end{equation}
and the conserved density $h_2$ is:
$$
\begin{array}{l}
\displaystyle{
h_2=
-\frac{c^2\epsilon_1}{2}\frac{{v''}^2}{v^5}
+\frac{15c^2\epsilon_1}{8}\frac{{v'}^4}{v^7}
-\frac{c\gamma_1}{2}\frac{{v'}^3}{v^5}
+\frac{c_1}{2}\frac{{v'}^2}{v^3}
-\frac{3c\beta_1}{2}\frac{{v'}^2}{v^4}
} \\
\displaystyle{
+\frac{5c\epsilon_1}{2}\frac{{v'}^2}{v^5}
-c\epsilon_1\frac{{v'}{u'}}{v^4}
+\gamma_1\frac{{v'}{u}}{v^3}
-\frac{\alpha_1}{2v} 
+\beta_1\frac{u}{v^2}
-\frac{\epsilon_1}{2}\frac{u^2}{v^3}
\,.}
\end{array}
$$

{\bf Case A2:} $\epsilon\neq0$,
and $\mc V=\mb F[u,v,u',v',u'',v'',\dots,Q^{-\frac12}]$.
In this case, the first non-zero equation is:
\begin{equation}\label{turhan:eq2}
\begin{array}{rcl}
\displaystyle{
\frac{du}{dt_2}
}
&=&
\displaystyle{
\vphantom{\Big(}
Q^{-\frac{3}{2}}
\Big( 
\gamma (\gamma\gamma_1+\epsilon c_1) {v^{(iv)}} 
+ \epsilon c_1 {v'''} {v} 
+ \gamma^2 \beta_1 {v'''}
- \epsilon(\gamma\gamma_1+\epsilon c_1) {u'''} 
} \\
&&
\displaystyle{
\vphantom{\Big(}
+ (\gamma\gamma_1+3\epsilon c_1){v''}{v'}
- \gamma_1 {v''}{v}^2 
+ \gamma\epsilon\alpha_1 {v''}
+ \epsilon\gamma_1 {u''} {v} 
- \gamma\epsilon\beta_1 {u''} 
} \\
&&
\displaystyle{
\vphantom{\Big(}
-3 \gamma_1 {v'}^2 {v}
+ \gamma\beta_1{v'}^2  
+2 \epsilon\gamma_1 {v'} {u'}
- \beta_1  {v'} {v}^2 
+ \epsilon\alpha_1 {v'} {v} 
+ \epsilon\beta_1 {u'} {v}
} \\
&& 
\displaystyle{
\vphantom{\Big(}
- \epsilon^2\alpha_1 {u'}   
\Big)
+\frac{3}{2}Q^{-\frac{5}{2}}\Big(
3 \gamma(\gamma\gamma_1+\epsilon c_1) 
\big( \gamma {v'''} {v''}
+{v'''} {v'} {v}
-\epsilon {v'''}{u'} \big)
} \\
&&
\displaystyle{
\vphantom{\Big(}
+\gamma(2\gamma\gamma_1+3\epsilon c_1){v''}^2{v}
+\gamma^3\beta_1 {v''}^2
-3\gamma(\gamma\gamma_1+\epsilon c_1)\big(
 \epsilon {v''}{u''} \!-\! {v''}{v'}^2 \big)
} \\
&&
\displaystyle{
\vphantom{\Big(}
+ (\gamma\gamma_1+3\epsilon c_1){v''}{v'}{v}^2
+2 \gamma^2\beta_1 {v''}{v'}{v}
- \epsilon (\gamma  \gamma_1+3\epsilon c_1){v''} {u'}{v}
} \\
&&
\displaystyle{
\vphantom{\Big(}
-2 \gamma^2\epsilon\beta_1 {v''}{u'}
-3(\gamma\gamma_1+\epsilon c_1)
\big(
\epsilon {u''}{v'}{v}
-\epsilon^2 {u''}{u'}
-{v'}^3 {v}
+ \epsilon {v'}^2{u'}\big)
} \\
&&
\displaystyle{
\vphantom{\Big(}
- \gamma_1{v'}^2{v}^3
+ \gamma\beta_1{v'}^2{v}^2
+2 \epsilon\gamma_1{v'}{u'}{v}^2
-2 \gamma\epsilon\beta_1{v'}{u'}{v}
- \epsilon^2\gamma_1{u'}^2{v}
} \\
&&
\displaystyle{
\vphantom{\Big(}
+ \gamma\epsilon^2\beta_1 {u'}^2  
\Big) 
+\frac{15}{4}Q^{-\frac{7}{2}}\big(\gamma\gamma_1+\epsilon c_1\big)
\Big( 
\gamma^3 {v''}^3
+3\gamma^2{v''}^2{v'}{v}
} \\
&&
\displaystyle{
\vphantom{\Big(}
-3\gamma^2\epsilon {v''}^2{u'}
+3\gamma {v''}{v'}^2{v}^2
-6\gamma\epsilon{v''}{v'}{u'}{v}
+3\gamma\epsilon^2{v''}{u'}^2
+{v'}^3{v}^3
} \\
&&
\displaystyle{
\vphantom{\Big(}
-3\epsilon {v'}^2{u'}{v}^2
+3\epsilon^2 {v'}{u'}^2{v}
- \epsilon^3 {u'}^3   
\Big) 
- 2Q^{-\frac{1}{2}}
\Big(
\gamma_1 {v''} + \beta_1 {v'} 
\Big)\,,
} \\
\displaystyle{
\vphantom{\Big(}
\frac{dv}{dt_2} 
}
&=&
\displaystyle{
\vphantom{\Big(}
Q^{-\frac{3}{2}}
\Big( 
\gamma(\gamma\epsilon_1-\epsilon\gamma_1){v'''} 
-\epsilon\gamma_1 {v''}{v}
+\gamma\epsilon\beta_1 {v''}
- \epsilon(\gamma\epsilon_1-\epsilon\gamma_1){u''}
} \\
&&
\displaystyle{
\vphantom{\Big(}
+(\gamma\epsilon_1-\epsilon\gamma_1){v'}^2
-\epsilon_1 {v'}{v}^2
+\epsilon\beta_1{v'}{v} 
+\epsilon\epsilon_1{u'}{v}
-\epsilon^2\beta_1{u'}   
\Big)
} \\
&&
\displaystyle{
\vphantom{\Big(}
+\frac{3}{2}Q^{-\frac{5}{2}}
\big(\gamma\epsilon_1-\epsilon\gamma_1 \big)
\Big( 
\gamma^2{v''}^2
+2\gamma{v''}{v'}{v}
-2\gamma\epsilon{v''}{u'}
+ {v'}^2{v}^2
} \\
&&
\displaystyle{
\vphantom{\Big(}
-2\epsilon{v'}{u'}{v}
+\epsilon^2{u'}^2  
\Big)
-2 Q^{-\frac{1}{2}}
\Big(
\epsilon_1 {v'}
\Big)
\,,}
\end{array}
\end{equation}
and the conserved density $h_2$ is
$$
\begin{array}{l}
\displaystyle{
h_2
=
\frac{1}{4}Q^{-\frac{5}{2}}
\Big(
\frac{1}{4}Q^{\frac{-5}{2}}
\Big(
\gamma^2(\epsilon c_1-\frac{\gamma^2}{\epsilon}\epsilon_1+2\gamma\gamma_1){v''}^2
+(\frac{\gamma^2}{\epsilon}\epsilon_1 - 2\gamma\gamma_1-\epsilon c_1) {v'}^2 {v}^2
} \\
\displaystyle{
-2(\gamma^2\epsilon_1-2\epsilon\gamma\gamma_1-\epsilon^2 c_1){v'}{u'}{v}
+\epsilon(\gamma^2\epsilon_1-2\epsilon\gamma\gamma_1-\epsilon^2 c_1) {u'}^2 
\Big) 
} \\
\displaystyle{
+
\frac{1}{3}Q^{\frac{-3}{2}}
\Big(
(\epsilon^2 c_1
-\gamma^2\epsilon_1+2\epsilon\gamma\gamma_1) {u''}
-\epsilon c_1 {v''}{v}
+(\frac{\gamma^2 }{\epsilon}\epsilon_1 - 2\gamma\gamma_1-\epsilon c_1) {v'}^2
} \\
\displaystyle{
+ (2\gamma_1-\frac{\gamma}{\epsilon} \epsilon_1){v'}{v}^2
+(\gamma\epsilon_1-2\epsilon\gamma_1){u'}{v}  
\Big) 
+
\frac{1}{3} Q^{\frac{-1}{2}}
\Big(
(10\gamma_1-8\frac{\gamma}{\epsilon}\epsilon_1) {v'}
} \\
\displaystyle{
-3\frac{\epsilon_1}{\epsilon} {v}^2
+6\beta_1 {v}
-3\epsilon\alpha_1 
\Big)
+
\frac{2}{3}\frac{\epsilon_1}{\epsilon} Q^{\frac{1}{2}}
\,.}
\end{array} 
$$

\begin{remark}\label{20140502:rem2}
Since in Case A1 the algebra of differential functions $\mc V[\log v]$ is normal
(where $\mc V$ is as in \eqref{20140416:eq3a}),
all conserved densities $h_n$ can be chosen in this algebra.
It is not difficult to show that actually they can be chosen in $\mc V$.
We conjecture 
that in Case A2 all conserved densities can be chosen in $\mc V$ as well
(where $\mc V$ is as in \eqref{20140416:eq3b}).
\end{remark}

\section{Bi-Hamiltonian hierarchies corresponding to  $H_0$ of type \eqref{20140404:eq9}}
\label{sec:5}

In this section we study the applicability of the Lenard-Magri scheme
in case B, when $H_0=H^{(c,\alpha,\beta,\gamma,\epsilon)}$
is as in \eqref{20140404:eq9}, with $c,\alpha,\beta,\gamma,\epsilon\in\mb F$,
such that $(p,q)\neq(0,0)$
(by Corollary \ref{20140408:cor2}, exactly in these cases $H_0$ is strongly skew-adjoint).
The case when both Poisson structures $H_0$ and $H_1$
have constant coefficients is not interesting 
(the Lenard-Magri scheme in this case is confined within quasiconstants).
Hence, we only consider the case when $H_1=H_{(c_1,\gamma_1,\epsilon_1)}$,
is as in \eqref{20140404:eq8}, with $c_1,\gamma_1,\epsilon_1\in\mb F$.

Recall that the space of Casimirs $C(H_0)$ is a Lie algebra with respect 
to the Lie bracket $\{\cdot\,,\,\cdot\}_1$.
We can use Lemma \ref{20140410:lem} to describe this Lie algebra explicitly.
\begin{lemma}\label{20140410:lemb}
Let $p\neq0$.
Consider the basis elements $\tint f_i$, $i=1,2,3,4$, of $C(H_0)$, 
as in Lemma \ref{20140410:lem}.
Then $\tint f_1$ is a central element of the Lie algebra $(C(H_0),\{\cdot\,,\,\cdot\}_1)$,
and the Lie brackets among all other basis elements are as follows:
\begin{enumerate}[(i)]
\item
If $q\epsilon\neq0$, then
$$
\begin{array}{l}
\displaystyle{
\{\tint f_2,\tint f_3\}_1=-\sqrt{\frac{q}{p}}\tint f_4
\,\,,\,\,\,\,
\{\tint f_2,\tint f_4\}_1=\sqrt{\frac{q}{p}}\tint f_3
\,,} \\
\displaystyle{
\{\tint f_3,\tint f_4\}_1
=-\sqrt{\frac{q}{p}}\frac{\beta}{\epsilon}\tint f_1
+\sqrt{\frac{q}{p}}\tint f_2
\,.}
\end{array}
$$
\item
If $\epsilon=0,q\neq0$, then
$$
\begin{array}{l}
\displaystyle{
\{\tint f_2,\tint f_3\}_1
=
\frac{\beta}{\gamma}\tint f_3
\,\,,\,\,\,\,
\{\tint f_2,\tint f_4\}_1
=
-\frac{\beta}{\gamma}\tint f_4
\,,} \\
\displaystyle{
\{\tint f_3,\tint f_4\}_1
=
-2\frac{\beta}{\gamma}\tint f_1
\,.}
\end{array}
$$
\item
If $\epsilon\neq0$, $q=0$, then
$$
\begin{array}{l}
\displaystyle{
\{\tint f_2,\tint f_3\}_1
=
-\frac{\gamma}{\epsilon}\tint f_1+2\tint f_4
\,\,,\,\,\,\,
\{\tint f_2,\tint f_4\}_1
=
\frac{\beta}{\epsilon}\tint f_1-\tint f_2
\,,} \\
\displaystyle{
\{\tint f_3,\tint f_4\}_1
=
\tint f_3
\,.}
\end{array}
$$
\item
If $\epsilon=q=0$, then
$$
\begin{array}{l}
\displaystyle{
\{\tint f_2,\tint f_3\}_1
=
-\tint f_2+\frac{\alpha}{\gamma}\tint f_4
\,\,,\,\,\,\,
\{\tint f_2,\tint f_4\}_1
=
\tint f_1
\,,} \\
\displaystyle{
\{\tint f_3,\tint f_4\}_1
=
-\tint f_4
\,.}
\end{array}
$$
\end{enumerate}
Consequently, 
the center of the Lie algebra $C(H_0)$
is spanned by $\tint v$.
\end{lemma}
\begin{proof}
Straightforward.
\end{proof}
\begin{remark}\label{20140424:rem}
If $p=0,q\neq0$,
then by Lemma \ref{20140410:lem}
the Lie algebra $C(H_0)$ is 2-dimensional abelian.
\end{remark}

By Theorem \ref{20140403:thm}(a),
for $\tint h_0\in C(H_0)$ the first step \eqref{20140404:eq1} of Lenard-Magri scheme 
has a solution $\tint h_1\in\mc V/\partial\mc V$
if and only if $\tint h_0$ lies in the center of the Lie algebra $(C(H_0),\{\cdot\,,\,\cdot\}_1)$.
Therefore, if $p\neq0$, by Lemma \ref{20140410:lemb},
we should start the Lenard-Magri scheme with $\tint h_0=\tint v$.
Since $\delta v$ lies in the kernel of $H_1(\partial)$,
by \eqref{20140404:eq2} we get that $\tint h_1\in C(H_0)$,
and therefore, by Theorem \ref{20140403:thm}(a), 
we can proceed to the next step only if $\tint h_1\in\mb F\tint v$.
Hence, in this case, the Lenard-Magri scheme is confined in $\mb F\tint v$ at each step
and it does not produce any integrable hierarchy.

Let us then consider the case when $p=0$, $q\neq0$.
We take $\mc V=\mb F[u,v,u',v',\dots]$.
We can choose in this case $\tint h_0=\tint v$ and $\tint h_1=\tint u$
and, by Corollary \ref{20140409:cor},
there exists an infinite sequence $\{\tint h_n\}_{n\in\mb Z_+}$
satisfying the Lenard-Magri recursive conditions \eqref{20140404:eq3}
for each $n\in\mb Z_+$.
We thus get an integrable hierarchy of bi-Hamiltonian equations as in \eqref{20140409:eq3},
provided that the integrals of motion $\tint h_n$ are linearly independent.

Linear independence in this case is easy to prove by considering the total polynomial degree.
Indeed, 
the first terms in the recurrence $H_0\delta h_n=H_1\delta h_{n-1}$ are
$\delta h_0=\left(\begin{array}{l} 0 \\ 1 \end{array}\right)$,
$\delta h_1=\left(\begin{array}{l} 1 \\ 0 \end{array}\right)$
and
$$
\delta h_2
=\left(\begin{array}{l} 
\epsilon u-\gamma v'-\beta v \\ 
-\beta u+\alpha v+\gamma u'+c v''
\end{array}\right)
\,,
$$
of total degrees $0$, $0$, and $1$ respectively.
Note that the homogeneous component of $\delta h_n=\left(\begin{array}{l} f_n \\ g_n \end{array}\right)$ 
of highest total degree satisfies the recurrence equation
$$
H_0(\partial)
\left(\begin{array}{l} 
f_n^0 \\ 
g_n^0 
\end{array}\right)
=
\left(\begin{array}{ll} 
u'+2u\partial  & v\partial \\
 \partial\circ v & 0 
\end{array}\right)
\left(\begin{array}{l}
f_{n-1}^0 \\ 
g_{n-1}^0
\end{array}\right)
\,.
$$
Since $H_0$ is a non-degenerate constant coefficients matrix differential operator,
it follows that, up to adding constant multiples of $\left(\begin{array}{l} 0 \\ 1 \end{array}\right)$
(namely elements of the kernel of $H_1$ in $\mb F[u,v,u',v',\dots]^2$), 
$f_n^0$ and $g_n^0$ are homogeneous of degree $n-1$.
In conclusion, 
$\delta h_n=\left(\begin{array}{l} f_n \\ g_n \end{array}\right)$ 
has both components of total polynomial degree $n-1$.
Therefore, the elements $\delta h_n$, $n\in\mb Z_+$, are linearly independent.

The first non-zero equation of the hierarchy starting with $\tint h_1=\tint u$, is $x$-translation symmetry:
$\frac{d{\bf u}}{dt_2}=\left(\begin{array}{c}u' \\ v'\end{array}\right)$.
We compute explicitly the conserved density $h_2$ 
and the corresponding first non-trivial equation of the hierarchy
$\frac{d{\bf u}}{dt_{3}}=H_1(\partial)\delta h_2$,
for all values of the parameters 
$c,\epsilon,\gamma,\alpha,\beta,c_1,\gamma_1,\epsilon_1 \in \mb F$ 
satisfying $p\,(=c\epsilon+\gamma^2)=0$ and $q\,(=\epsilon\alpha-\beta^2) \ne0$. 
We have
$$
h_2
=
-\frac{c}{2}{v'}^2
-\gamma v'u
+\frac{\alpha}{2}v^2
-\beta uv 
+\frac{\epsilon}{2}u^2
\,,
$$
and the first non-trivial equation is
\begin{equation}\label{turhan:eq4}
\begin{array}{rcl}
\displaystyle{
\frac{du}{dt_3}
}
&=&
\displaystyle{
\vphantom{\Big(}
(c\gamma_1-c_1\gamma){v^{(iv)}}
+cv'''v
-\beta c_1 v'''
+(\epsilon c_1+\gamma\gamma_1)u'''
} \\
&&
\displaystyle{
\vphantom{\Big(}
-2\gamma v''u
+\alpha\gamma_1 v''
+\gamma u''v
-\beta\gamma_1 u''
-\gamma v'u'
+\alpha v'v
} \\
&&
\displaystyle{
\vphantom{\Big(}
-2\beta v'u
-2\beta vu'
+3\epsilon u'u
\,,} \\
\displaystyle{
\vphantom{\Big(}
\frac{dv}{dt_3} 
}
&=&
\displaystyle{
\vphantom{\Big(}
(c\epsilon_1+\gamma\gamma_1)v'''
-\gamma v''v
+\beta\gamma_1 v''
+(\gamma\epsilon_1-\epsilon\gamma_1)u''
-\gamma {v'}^2
} \\
&&
\displaystyle{
\vphantom{\Big(}
-2\beta v'v
+\epsilon v'u
+\epsilon vu'
+\alpha\epsilon_1 v'
-\beta\epsilon_1 u'
\,.}
\end{array} 
\end{equation}

\begin{remark}\label{turhan:rem1}
The $(\epsilon_1,\gamma_1)=(0,0)$ subcase of $H_1$ in the present section 
coincides with $H_0$ of the Case A1 in Section \ref{sec:4}.
So for these particular values of the parameters we have two-sided Lenard-Magri scheme, 
the first side of which is the above 
one restricted with $(\epsilon_1,\gamma_1)=(0,0)$, and the second side is the one in Case A1
of Section \ref{sec:4} (with subscripted parameters unsubscripted and unsubscripted parameters 
subscripted to be consistent with the present section).
\end{remark}
\begin{remark}\label{turhan:rem2}
For $\epsilon_1\neq0$, $H_1$ of the present section is $H_0$ in Case A2
of Section \ref{sec:4}. 
Letting, as in case Case A2 of Section \ref{sec:4},
$c=\frac{-\gamma^2}{\epsilon}$, we again get a two-sided
Lenard-Magri scheme. 
\end{remark}
%

\section{Concluding remarks}
\label{sec:6}
The equation in Case A1 of Section~\ref{sec:4} is triangular 
for $\epsilon_1=0$, 
in case of which a solution reduces to succesive solutions of scalar equations. 
For $\epsilon_1\ne0$ the equation in the variables $w,v$ 
where $w=u-\frac{\beta_1}{\epsilon_1}v-\frac{\gamma_1}{\epsilon_1}v'$ becomes:
\begin{equation}\label{6.1}
\begin{array}{rcl}
\displaystyle{
\frac{dw}{dt} 
}
&=&
\displaystyle{
\frac{1}{\epsilon_1}\Big(
\big(c_1\epsilon_1+\gamma_1^2\big)
\Big(\frac1v\Big)^{\prime\prime\prime}
+\big(\alpha_1\epsilon_1-\beta_1^2\big)
\Big(\frac1v\Big)^{\prime}
\Big),} \\
\displaystyle{
\frac{dv}{dt} 
}
&=&
\displaystyle{
-c\epsilon_1
\frac{1}{v}\Big(\frac1v\Big)^{\prime\prime\prime}
-\epsilon_1\Big(
\frac{w}{v^2}
\Big)^\prime
\,.}
\end{array} 
\end{equation}
For $c=0$ this is, 
up to a renaming of coefficients, the Harry Dym type equation considered in \cite{AF88,BDSK09}.
For $c\neq0$ this equation seems to be new.


The equation in Case A2 of Section~\ref{sec:4} in the variables $w,v$ where $w=Q^{-\frac12}$, becomes:
\begin{equation}\label{6.2}
\begin{array}{rcl}
\displaystyle{
\frac{dw}{dt} 
}
&=&
\displaystyle{
(\gamma^2\epsilon_1 -2\epsilon\gamma\gamma_1-\epsilon^2c_1) w^3w'''
+(\epsilon\gamma_1-\gamma\epsilon_1) (w^4v''+2w^3w'v')}
\\
&&
\displaystyle{
-\epsilon_1 (w^4vv'+ w^3w'v^2)
+\epsilon\beta_1 (w^4v'+2w^3w'v)
-\epsilon^2\alpha_1 w^3w',
}
\\
\displaystyle{
\frac{dv}{dt} 
}
&=&
\displaystyle{
2(\gamma\epsilon_1-\epsilon\gamma_1) w''-2\epsilon_1(w'v+wv') +2\epsilon\beta_1w'
\,.}
\end{array}
\end{equation} 
This equation seems to be new for 
$\gamma^2\epsilon_1 -2\epsilon\gamma\gamma_1-\epsilon^2c_1\neq0$.
In the case 
$\gamma^2\epsilon_1 -2\epsilon\gamma\gamma_1-\epsilon^2c_1=0$, 
the remaining second order equation is an equation considered in \cite{MSY87}. 
Further reduction with $\gamma_1=\frac{\epsilon_1}{\epsilon}\gamma$, 
gives a first order equation which corresponds to those considered in \cite{GN90}.

Recall that in Section~\ref{sec:5} the Lenard-Magri scheme works
if $p\,(=c\epsilon+\gamma^2)=0$ and $q\,(=\alpha\epsilon-\beta^2)\neq0$.
If $c=\gamma=0$, $\epsilon c_1\neq0$, the equation in the variables 
$w=\frac{1}{q} \big(\epsilon^2c_1u -\beta\epsilon c_1 v +\frac13 \beta^2\epsilon_1c_1\big)$,
$z=\frac{1}{\sqrt{q}} \big(\epsilon c_1 v - \beta\epsilon_1 c_1 \big)$, 
$\tau=\frac{\epsilon^2c_1^2}{q^{3/2}} t$, 
$y=\frac{1}{\epsilon c_1\sqrt{q}} x$,
becomes:
\begin{equation}\label{turhan:eq}
\begin{array}{rcl}
\displaystyle{
\frac{dw}{d\tau} 
}
&=&
\displaystyle{
w'''+3ww'+zz'+\gamma_1z'',
}
\\
\displaystyle{
\frac{dz}{d\tau} 
}
&=&
\displaystyle{
(wz)'+Kz'-\gamma_1w''
\,,}
\end{array}
\end{equation} 
where $K=\frac{1}{3q}\epsilon_1c_1(3\alpha\epsilon-4\beta^2)$, 
and primes on $w$ and $z$ denote derivative with respect to $y$.  
The parameters $K$ and $\gamma_1$ are essential parameters, 
i.e. they cannot be removed by rescaling of the variables in the equation.
%
The 2-parameter family of equations \eqref{turhan:eq}, which seems to be new,
turns into the Fokas-Liu equation \cite{FL96} for $\gamma_1=0$,
into the Kupershmidt equaiton \cite{Kup85b} for $K=0$,
and into the Ito equation \cite{Ito82} for $\gamma_1=K=0$.
 
In the cases $c=0$, $\epsilon\ne0$, $c_1=0$, and  $c=0$, $\epsilon=0$, upon a potentiation $u=w'$, 
equation \eqref{turhan:eq4} becomes second order of type considered in \cite{MSY87}. 
The case $c=0$, $\epsilon=0$, $c_1=0$ corresponds, by a similar transformation as above, 
to the Kaup-Broer equation \cite{Kup85a}. 

%
For $c\neq0$, and 
$\gamma(c\epsilon_1+\gamma\gamma_1)(c^2\epsilon_1+2c\gamma\gamma_1-c_1\gamma^2)\neq0$,
we transform equation \eqref{turhan:eq4} by the change of variable
$u=m_4w'-m_2 z', v=-m_3 w+m_1 z$, 
and rescaling $y=k_x x$,  
$\tau=k_t t$ with constant transformation parameters 
$m_1=1,m_2=\frac{c}{\gamma},m_3=1,
m_4=\frac{\gamma c_1-c\gamma_1 }{c\epsilon_1+\gamma\gamma_1},
k_x=\frac{-\gamma}{c\epsilon_1+\gamma\gamma_1},
k_t=\frac{c(c\epsilon_1+\gamma\gamma_1)^4}
{\gamma^2(c^2\epsilon_1 +2c\gamma\gamma_1 -\gamma^2 c_1)^2}$.
As a result, we obtain:
\begin{equation}\label{turhan:eqb}
\begin{array}{rcl} 
\displaystyle{
\frac{dw}{d\tau}
} &=& \displaystyle{
w'''
+\tilde\beta z''
-\frac32{w'}^2
+2\tilde\beta ww'
-2\tilde\beta w'z
+\tilde \alpha w'
-\tilde \alpha z'
+\frac12\tilde\alpha(w-z)^2
\,,}
\\
\displaystyle{
\frac{dz}{d\tau}
} &=& \displaystyle{
w''w
-w''z
+K\tilde\beta w''
-\frac12{w'}^2
-w'z'
+K\tilde\alpha w'
+2\tilde\beta w z'
} \\
&&
\displaystyle{
-2\tilde\beta zz'
-K\tilde\alpha z'
+\frac12\tilde\alpha(w-z)^2
\,,}
\end{array} 
\end{equation}
where
$\tilde\alpha=\frac{c(c\epsilon_1+\gamma\gamma_1)^4}
{\gamma^2(c^2\epsilon_1 +2c\gamma\gamma_1 -\gamma^2c_1)^2} \alpha$,  
$\tilde \beta=\frac{c(c\epsilon_1+\gamma\gamma_1)^2 }
{\gamma^2(c^2\epsilon_1 +2c\gamma\gamma_1 -\gamma^2c_1)} \beta$, 
$K=\frac{\gamma^2(c_1\epsilon_1+\gamma_1^2)}{(c\epsilon_1+\gamma\gamma_1)^2}$.
In the above equation, if $c_1\epsilon_1+\gamma_1^2=0$, 
switching the roles of $H_0$ and $H_1$ we get a two-sided Lenard-Magri scheme,
according to Remarks \ref{turhan:rem1} and \ref{turhan:rem2}, 
depending on whether $\epsilon_1=0$ or $\epsilon_1\neq0$, respectively. 
In the former case ($\epsilon_1=\gamma_1=0$)
the second side is related to the Hary Dym type of equation. 
In the latter case ($\epsilon_1\gamma_1\neq0$), and in the case $K\neq0$,
equation \eqref{turhan:eqb} seems to be new.
 
For $c\neq0$, $\gamma=0$, $\epsilon_1\neq 0$, the same transformation as above with  
$m_1=0$,
$m_2=c\epsilon_1$,
$m_3=c\epsilon_1$,
$m_4=-c\gamma_1$,
$k_x=-c$,
$k_t=\frac{1}{c^6\epsilon_1^3}$
gives:
\begin{equation}\label{6.5}
\begin{array}{rcl}
\displaystyle{
\frac{dw}{d\tau} 
}
&=&
\displaystyle{
w''' + \tilde\beta z'' + 2\tilde\beta ww'+ \tilde\alpha w',
}
\\
\displaystyle{
\frac{dz}{d\tau} 
}
&=&
\displaystyle{
w''w
+K\tilde\beta w''
-\frac12{w'}^2
+2\tilde\beta wz'
+\frac12\tilde\alpha w^2
\,,}
\end{array}
\end{equation}
where 
$\tilde\alpha=\frac{\alpha}{c^3}$,  
$\tilde \beta=\frac{\beta}{c^2}$, 
$K=\frac{c_1\epsilon_1+\gamma_1^2}{\epsilon_1^2}$. Since $\epsilon_1\neq0$ here, we have a second side of Lenard-Magri scheme by Remark~\ref{turhan:rem2} if $K=0$. 
This equation seems to be new. 

For $c\neq0$,  $\epsilon_1=\frac{\gamma\gamma_1}{c}$, $\gamma(c\gamma_1-c_1\gamma)\neq0$ 
the transformation with  
$m_1=\frac{\gamma}{c}(\gamma c_1-c\gamma_1)$,
$m_2=\gamma c_1-c\gamma_1$,
$m_3=0$,
$m_4=\frac{\gamma}{c(\gamma c_1-c\gamma_1)}$,
$k_x=\frac{\gamma^2}{c} $,
$k_t=\frac{c^6}{\gamma^9 (c\gamma_1-\gamma c_1)^3}$
gives:
\begin{equation}\label{6.6}
\begin{array}{rcl}
\displaystyle{
\frac{dw}{d\tau} 
}
&=&
\displaystyle{
w''' 
+ \tilde\beta z'' 
+ \frac32{w'}^2 
+ 2\tilde\beta zw' 
-\frac12\tilde\alpha z^2 ,
}
\\
\displaystyle{
\frac{dz}{d\tau} 
}
&=&
\displaystyle{
w''w
+K\tilde\beta w''
+w'z'
+2\tilde\beta zz'
-K \tilde\alpha z'
\,,}
\end{array}
\end{equation}
where 
$\tilde\alpha=\frac{c^3}{\gamma^6}\alpha$,  
$\tilde \beta=\frac{c^2}{\gamma^4}$, 
$K=\frac{\gamma^2\gamma_1}{c(c\gamma_1-\gamma c_1)}$.

In the case $c\neq0$, $\epsilon_1=\frac{-c_1\gamma^2}{c^2}$, $\gamma_1=\frac{\gamma c_1}{c}$, upon potentiation $u=w'$, equation \eqref{turhan:eq4}
becomes a second order equation of type considered in \cite{MSY87}.  


\end{document}